\documentclass[pdflatex,sn-mathphys-ay,iicol]{sn-jnl}

\usepackage{bm}
\usepackage{xspace}
\usepackage{geometry}
\usepackage{array}
\usepackage{subfigure}
\usepackage{epsfig}
\usepackage{hhline}
\usepackage{natbib}
\usepackage{bbm}
\usepackage[T1]{fontenc}
\usepackage[]{graphicx,float,latexsym,times}
\usepackage{amsfonts,amstext,amsmath,amssymb,amsthm}
\usepackage{algorithm}%
\usepackage{algorithmicx}%
\usepackage{algpseudocode}%
\usepackage{listings}%

\usepackage{booktabs}     
\usepackage{multirow}     
\usepackage{multicol}     
\usepackage{makecell}     
\usepackage{graphicx}     
\usepackage{arydshln}     
\usepackage{caption}    
\DeclareCaptionLabelFormat{algnonumber}{Algorithm}
\usepackage[utf8]{inputenc}    
\DeclareMathAccent{\widehat}{\mathord}{largesymbols}{"62}
\DeclareMathAccent{\widetilde}{\mathord}{largesymbols}{"65}

\def\sgn{\mathop{\mathrm{sgn}}}
\newcolumntype{B}{>{\usefont{T1}{ptm}{m}{n}\fontsize{.3cm}{.3cm}\selectfont }c} 

\def\eb*{\textrm{\mathversion{bold}$\mathbf{\beta^*}$\mathversion{normal}}}

\def\oo{\textrm{\mathversion{bold}$\mathbf{0}$\mathversion{normal}}}
\def\eb{\textrm{\mathversion{bold}$\mathbf{\beta}$\mathversion{normal}}}

\def\eU{\textrm{\mathversion{bold}$\mathbf{\Upsilon}$\mathversion{normal}}} 

\def\eE{\mathbb{E}}
\def\eP{I\!\!P}
\def\e1{1\!\!1}

\theoremstyle{plain}

\newtheorem{theorem}{Theorem}[section]

\newcommand{\beqn}{\begin{eqnarray*}}
\newcommand{\eeqn}{\end{eqnarray*}}
  
\def\ee1{\textrm{\mathversion{bold}$\mathbf{\varepsilon}$\mathversion{normal}}}

\def\eu{\mathbf{{u}}}
\def\ez{\mathbf{{z}}}

\newcommand{\N}{\mathbb{N}}
\newcommand{\R}{\mathbb{R}}
\newcommand{\PP}{\mathbb{P}}

\def\eX{\mathbf{X}}

\newcommand{\bfW}{{\bf W}}

\newcommand{\Var}{\mathbb{V}\mbox{ar}\,}

\def\argmin{\mathop{\mathrm{arg\,min}}} 
\def\oo{\textrm{\mathversion{bold}$\mathbf{0}$\mathversion{normal}}}
\renewcommand{\algorithmicrequire}{\textbf{Input:}}
\renewcommand{\algorithmicensure}{\textbf{Output:}}

\begin{document}

\title {{\bf Transfert   learning and adaptive   LASSO quantile}}
 
 \author*{\fnm{Gabriela} \sur{Ciuperca}}\email{gabriela.ciuperca@univ-lyon1.fr}  
 \affil{\orgname{Université Lyon 1}, \orgdiv{CNRS, ICJ, UMR 5208}, \orgaddress{
 		\city{Villeurbanne}, 
 		\country{France}}}
 
 \abstract{We propose for a quantile regression an estimation method for transferring knowledge using two $L_1$ penalties based on an estimator obtained from a source database. The proposed transfer learning estimator satisfies the properties of consistency and sparsity. Its convergence rate and asymptotic behavior are studied in several scenarios. This knowledge transfer results in a shorter computation time than that of the standard adaptive  LASSO estimator. Another advantage of our method is that it can be applied to models with non-Gaussian errors. In addition,
 in order to implement the computing of   the adaptive transfer LASSO quantile estimator, we propose an algorithm. The simulations confirm the theoretical results and demonstrate that the adaptive learning estimator, calculated using the proposed algorithm, is more competitive than the LASSO estimators. Finally, we illustrate the practical utility of the proposed transfer learning estimator and algorithm using a real-data application involving the physicochemical properties of protein tertiary structures. }



 \maketitle


\section{Introduction}
\label{section_Introduction}
A major challenge in practical applications is the automatic selection of significant explanatory variables a high-dimensional model with non-Gaussian errors without resorting to hypothesis testing. In addition, estimation methods and the algorithms used to implement them must have good properties without taking too long to compute. Hence the idea of extracting some of the information from the source data and refining the estimator on a target dataset.
Therefore, in this paper, we consider a linear model in which the errors are not necessarily Gaussian and there may be a large number of explanatory variables. In recent years, the adaptive LASSO method is one of the most commonly used approaches to automatically select significant variables in high-dimensional models. On the other hand, to address non-normal error distributions, which may include outliers, on option is to consider quantile estimation. Combining these two estimation techniques,  the adaptive  LASSO quantile method was developed (see for  example \cite{Zou.Yuan.08}, \cite{Zheng-Gallagher-Kulasekera-13}, \cite{FFB14}, \cite{Kaul-Koul.15}, \cite{Ciuperca-16}, \cite{Hou.Meng.Tian.26}). 
The non-differentiability of the quantile function, combined with an $L_1$-type penalty, leads to numerical concerns and increased computation time, especially if the model is high-dimensional on a very large dataset. That is why, in this paper, we will consider two databases: a source database, which will be  larger, and a target database, which will be the smaller of the two. Based on the source data, we will simply calculate the quantile estimator, which will then be used as knowledge  to the target data. This transfer allows us to maintain fast convergence because it draws on knowledge derived from a large database, whereas the target database only enables us to automatically select significant variables. In this way, using this approach, we then construct an adaptive transfer LASSO quantile estimator on the target data. First, we show that this estimator has the sparsity property, and  determine its asymptotic distribution in three scenarios. These results   allow us to determine the optimal configuration of the source and target data, as well as the tuning parameters, in order to obtain a more accurate estimator that converges more quickly. Next, we propose an algorithm to compute this transfer learning estimator.  It should be noted that the simulations we conducted allows us to conclude that the computation time is faster using the transfer learning algorithm than using the adaptive LASSO quantile estimator with both the target and source datasets.  Moreover, numerical simulations confirm the theoretical results obtained for the transfer learning estimator and demonstrate its advantages over traditional adaptive LASSO quantile estimators. This is a particularly interesting result when target data is scarce, since our method only requires unpenalised parameter estimation on the source model, which is more accessible in practice. \\
Similar estimation studies involving one or more  source data, and then transferring that knowledge to a target model have been studied  extensively in recent years. So, \cite{Li-Cai-ZLi.22} propose a transfer learning approach for estimating the target regression parameters in high-dimensional linear regression, based on $K$ source regressions. This algorithm is based on calculating an 'information set', which contains source features that satisfy a sparsity property. For a large-scale source model and limited target data, \cite{Gu-Han-Duan.25} propose an angle-based transfer learning method. These two papers present an estimation method based on the assumption that model errors are centered and have bounded variance. This method involves penalizing the sum of the squares of the errors using the $L_1$ or $L_2$ norm.  \cite {Takada.Fujisawa.24} consider i.i.d. Gaussian errors and propose a method that combines LASSO and transfer methods. This method involves penalizing the least squares with a LASSO-type penalty and an L1-type penalty on the contraction vector, which is the difference between an unknown parameter and an initial estimator constructed from the source data. The convergence rate and asymptotic behavior of the  adaptive transfer LASSO estimator are investigated for various configurations of source and target sets.  For i.i.d. sub-Gaussian errors with mean 0 and bounded variance, \cite{Qu.2025} proposes an automatic transfer learning method involving the assignment of adaptive penalties to various parameters in order to automatically assess transferability.\\
If the errors do not follow a normal distribution, quantile regression is a possible estimation method. \cite{Zhang.Zhu.25} propose a transfer learning algorithm based on convolution smoothing to improve the estimation and forecasting of a high-dimensional quantile model. The proposed method is entirely different from the one discussed in the present paper. The properties of the quantile learning estimator presented here pertain to its rate of convergence. When the target and source data are sufficiently similar, the empirical quantile loss with LASSO penalty is used in the paper \cite{Qiao-He-Zhou.24} to propose transfer learning algorithms that provide error bounds for estimators of the target regression coefficients. \cite{Huang.Wang.Wu.22}, \cite{Jin.Yan.Aseltine.Chen.24} have also studied transfer Lasso   quantile estimators, but these minimize a different process than the one discussed in this paper. Other research on transfer learning can be found in: \cite{Tian.Feng.23}, \cite{Li.Zhang.Cai.Li.24} for high-dimensional generalized linear models, \cite{Wang.Wang.Lian.24} for adaptive Huber regression, \cite{Liu.Song.25} for expectile regression.\\
The paper is organized as follows. Section \ref{section_modele} introduces the model, the assumptions and the adaptive transfer LASSO quantile estimator. Section \ref{section_resultats} presents the asymptotic properties of this estimator. An algorithm for computing numerically  this is provided in Section \ref{section_alg}. Section \ref{section_numeric} presents numerical results   to analyze the adaptive transfer LASSO quantile estimator, confirm the theoretical properties discussed in Section \ref{section_resultats}, and compare it with two adaptive quantile LASSO estimators. In Section \ref{section_application} we present an application on real data. The proofs of the theoretical results are relegated to Section \ref{section_proofs}.

\section{Model and estimation method}
\label{section_modele}
In this section we introduce the source and target models, the assumptions made them, and the transfer learning estimator.\\ 
We will start this section by introducing the general notation used throughout the paper.  First, all vectors and matrices are denoted by bold symbols and all vectors are written as  column vectors. For a vector $\textbf{v}$, we denote by $\textbf{v}^\top$ its transposed and by $\|\textbf{v} \|_2$ its Euclidean norm.  For a positive definite matrix $\textbf{M}$, we denote by $\lambda_{\min}(\textbf{M)}$  and $\lambda_{\max}(\textbf{M)}$ its the smallest and largest eigenvalues, respectively.  If $E$ is an event, then $\e1_E$ denotes the indicator function that the event $E$ occurs.  Given a set ${\cal S}$, we denote its cardinality by $|{\cal S}|$.
We will also use   the following notations: if $V_n $ and $U_n$ are random variable sequences, $V_n=o_{\eP}(U_n)$ means that $\lim_{n \rightarrow \infty} \eP[|U_n/V_n| > e]=0$ for any $e>0$, $V_n=O_{\eP}(U_n)$ means that there exists a finite $C>0$ such that $\eP[|U_n/V_n| > C]< e$ for any $n$ and $e$. If $(a_n)_{n \in \N}$ and $(b_n)_{n \in \N}$ are  deterministic sequences,   $a_n=o (b_n)$ means that the sequence $a_n/b_n \rightarrow 0$ for $n \rightarrow\infty$, 	and if both sequences are positive, we will also use the notation $b_n \gg a_n$. And finally, the notation  $a_n=O (b_n)$ means that the sequence $a_n/b_n $ is bounded for sufficiently large $n$. Notations $ \overset{\cal L} {\underset{n \rightarrow \infty}{\longrightarrow}}$, $ \overset{\PP} {\underset{n \rightarrow \infty}{\longrightarrow}}$ represent the convergence in distribution and in probability, respectively, as $n \rightarrow \infty$.  Throughout this paper, $c$ will always denote a generic constant, not depending on $n$ or $m$, and its value is not of interest. \\
In the context of transfer learning, we consider the following linear model on the source data:
 \begin{equation}
 	\label{eq1_source}
 	Y_i=\sum_{j=1}^{p} X_{ji}\beta_j+\varepsilon_i=\eX_i^\top \eb+\varepsilon_i, \quad i=1, \cdots, m.
 \end{equation}
We study the estimation of the following model  on the target data:
  \begin{equation}
 	\label{eq1_cible}
 	Y_i=\sum_{j=1}^{p} X_{ji}\beta_j+\varepsilon_i=\eX_i^\top \eb+\varepsilon_i, \quad i=m+1, \cdots, m+n.
 \end{equation}
The same dependent variable $Y$ and the same explanatory variables $\eX=(X_1, \cdots , X_p)$ were observed  in both models but the approach remains valid if the source data contains additional explanatory variables. Consequently, the true parameter $\eb^*=(\beta_1^*, \cdots, \beta^*_p)$, which is unknown, is the same for models (\ref{eq1_source}) and (\ref{eq1_cible}). Let's also consider the set ${\cal S}^*\equiv \{ j \in \{1, \cdots, p\} ; \beta^*_j \neq 0\}$ contains the indices of the non-zero coefficients. The number $p$ of variables is less than $n$.\\
The explanatory variables of (\ref{eq1_source})  and (\ref{eq1_cible}) are assumed to be deterministic while the model errors $(\varepsilon_i)_{1 \leqslant i \leqslant m+n}$ and the dependent variable $Y$ are random.\\
 We assume that $\lim_{(m,n)\rightarrow\infty} n m^{-1}=r_0\geq 0$, so, the target sample size $n$ cannot be much larger than  $m$.\\
In this paper, we do not make the common assumption in applications that the model’s errors follow a normal distribution.  One way to estimate the model’s parameters is to use the quantile method, which employs the quantile function as its loss function $\rho_\tau(r) = \tau r  \e1_{r>0}- (1-\tau)r \e1_{r \leq 0}$, $r \in \R$.\\
When there are a large number of explanatory variables, automatic selection of significant variables is a more appropriate approach than hypothesis testing. One of the techniques used for automatic variable selection is the adaptive LASSO penalty. This was introduced by \cite{Zou.06} for a least squares model and by \cite{Zou.Yuan.08} for composite quantile regression. When the number of observations in the target data is small, a natural approach to obtaining an estimator that satisfies the sparsity property—one that has the same properties as the estimator calculated using both the target and source data—is to first construct a consistent estimator based on the source data and then transfer and enhance those properties to an estimator constructed using the target data. 
Building on this idea, this paper proposes estimating the parameter $\eb$ of the target model using the source data while automatically selecting the relevant explanatory variables. 
To achieve this, we first calculate the quantile estimator on the source data:
 \begin{equation}
 	\label{eq2}
 	\widetilde \eb_m \equiv \argmin_{\eb \in \R^p}\sum^m_{i=1} \rho_\tau(Y_i-\eX_i^\top \eb).
 \end{equation}
 The components of  $\widetilde \eb_m$ are denoted $(\widetilde \beta_{m,1}, \cdots , \widetilde \beta_{m,p})$. \\
This estimator will serve as a transfer function for constructing the adaptive  transfer LASSO quantile estimator based on the target data:
 	\begin{multline}
 			\label{eq_cible}
 		\widehat \eb_{m,n} \equiv \argmin_{\eb \in \R^p} \bigg(\sum^{m+n}_{i=m+1} \rho_\tau(Y_i-\eX_i^\top \eb)\\
 		+\lambda_n \sum^p_{j=1} v_{m,j} |\beta_j| +\eta_n \sum^p_{j=1} \omega_{m,j} |\beta_j -\widetilde\beta_{m,j}| \bigg),
 	\end{multline}
with the random variables: 
\[
v_{m,j}=|\widetilde\beta_{m,j}|^{-\gamma_1}, \qquad \omega_{m,j}=|\widetilde\beta_{m,j}|^{\gamma_2},
\]
with $\gamma_1$ and $\gamma_2$ two known parameters $\geq 0$. 
For the sake of simplicity, unless otherwise specified, we will omit the subscript $m$ and denote $\widehat \eb_{m,n}$ by $\widehat \eb_{n}$ and its components by $(\widehat \beta_{n,1}, \cdots , \widehat \beta_{n,p})$.   This estimator has the advantage of automatically selecting zero coefficients, and, through the second penalty term, serves to transfer information from the source model to the target model.\\
Minimizing the expression on the right-hand side of (\ref{eq1_cible}) is equivalent to minimizing $\sum^{m+n}_{i=m+1} \rho_\tau(Y_i-\eX_i^\top \eb)$, subject to two constraints on $\eb$: $\sum^p_{j=1} v_{m,j} |\beta_j| \leq c_{1,n}$ and $\sum^p_{j=1} \omega_{m,j} |\beta_j -\widetilde\beta_{m,j}| \leq c_{2n}$, with $c_{1,n}$ and $c_{2,n}$ constants that depend on $n$.\\
Taking these two penalties/constraints into account is equivalent to finding, among the $(\beta_j)_{1 \leqslant j \leqslant p}$  values, the minimizer of the quantile process those that lie within the  region $\sum^p_{j=1} v_{m,j} |\beta_j| \leq c_{1,n}$ which will result in sparse estimators, and also within the  region $\sum^p_{j=1} \omega_{m,j} |\beta_j -\widetilde\beta_{m,j}| \leq c_{2n}$, which will encourage $\widehat \beta_{n,j}$ to be close to $\widetilde \beta_{m,j}$ and thus reduce bias.\\
For models (\ref{eq1_source}) and (\ref{eq1_cible}), the following three assumptions are considered. First, for model errors, we assume:
\begin{description}
	\item \textbf{(A1)} $(\varepsilon_i)_{1 \leq i \leq m+n}$ are i.i.d., with  the distribution function  $F$ and density function $f$. The density function $f$ is continuously,  strictly positive in a neighborhood of zero and has a bounded first derivative in the neighborhood of 0. The $\tau$th quantile of $\varepsilon_i$ is zero: $\tau= F(0)$. 
\end{description}  
The design $(\eX_i)_{1 \leqslant i \leqslant n}$ is such that: 
\begin{description}
	\item \textbf{(A2)} $n^{-1} \sum^n_{i=1} \eX_i \eX_i^\top {\underset{n \rightarrow \infty}{\longrightarrow}} \eU$, with $\eU$ a positive definite $p \times p$ matrix.  
	\item \textbf{(A3)} For  $(m,n) \rightarrow\infty$, $d^{-1}_{m,n} \max_{1 \leqslant i \leqslant m+n} \| \eX_i \|_2  \rightarrow 0$, with $d_{m,n}=\min(m^{1/2}, n^{1/2}, n\lambda_n^{-1})$.
\end{description}  
The elements of the matrix $\eU$ are denoted by $(\Upsilon_{ij})_{1 \leqslant i ,j \leqslant p}$.\\
Assumptions (A1) and (A2) are standard for a quantile linear regression (see \cite{Koenker-05}, \cite{Zou.Yuan.08}, \cite{Ciuperca-16}, \cite{Ciuperca.2019}). Furthermore, assumption (A3) is required in order to control the convergence rate of $\widehat \eb_n$.\\
To study the properties of $\widehat \eb_n$, let us consider  the parameter $\eb$ of the form  $\eb = \eb^* + l^{-1} \eu$, where $l \in \mathbb{R}^+$, $\eu \in \mathbb{R}^p$, and $\| \eu \|_2 \leq c < \infty$. In fact, $l$ is a sequence that may depend on $m$, $n$, and $\lambda_n$ such that $l \rightarrow\infty$ as $m$ and $n \rightarrow\infty$. Further on,  we will consider the following three cases for the sequence $l$, which depends on $m$ or $n$: $l=m^{1/2}$, $l=n^{1/2}$, or $l=n/\lambda_n$. \\
 Then, for $i=m+1, \cdots, m+n$, we can write: $Y_i=\eX_i^\top \eb^*+\varepsilon_i=\eX_i^\top \eb -l^{-1}\eX_i^\top \eu +\varepsilon_i$, from which we have:
\[
Y_i - \eX_i^\top \eb =\varepsilon_i- l^{-1}\eX_i^\top \eu.
\]

Under assumptions (A1) and (A2), according to \cite{Koenker-05}, the random vector $\widetilde \eb_m$ defined by (\ref{eq2}) is $\sqrt{m}$-consistent and asymptotically normal. Then,  let us consider the random vector  $\ez_m = \sqrt{m} (\widetilde \eb_m-\eb^*)$. Thus $\widetilde \eb_m=\eb^*+m^{-1/2} \ez_m$, with $\ez_m \in \R^p$, $\|\ez_m \|_2 \leq c < \infty$ with a probability converging to 1 when $m \rightarrow \infty$. From \cite{Koenker-05} we have under assumptions (A1)  and (A2) that   $\ez_m \overset{\cal L} {\underset{m \rightarrow \infty}{\longrightarrow}}  {\cal N}(\oo_p, \tau(1-\tau)f^{-2}(0) \eU^{-1}) \equiv \ez $.\\
Finally, let’s note that  the components of the vectors $\eu$ and $\ez_m$ are, respectively, $(u_1, \cdots , u_p)$ and $(z_{m,1}, \cdots, z_{m,p})$. 

\section{Asymptotic properties of $\widehat \eb_n$}
\label{section_resultats}
In this section, we will present the convergence rate, the asymptotic distribution, and the sparsity of the estimator $\widehat \eb_n$ in (\ref{eq1_cible}). The proofs of the two theorems presented below are given in   Section \ref{section_proofs}. \\
For this study, we first need to examine 
\begin{equation}
	\label{eq3}
	{\cal P}_n \equiv \frac{1}{n} \sum^{m+n}_{i=m+1} \big(\rho_\tau(\varepsilon_i- l^{-1}\eX_i^\top \eu)- \rho_\tau(\varepsilon_i)\big)
\end{equation}
which will be done using the results of Theorem 1 in  \cite{Ciuperca.2019} and Theorem 2.1 in \cite{Zou.Yuan.08}.\\
The random process $    {\cal P}_n$ can be written in the form:
\begin{equation}
	\label{eq4}
		{\cal P}_n=\eE[	{\cal P}_n]+\frac{1}{n} \bfW_n \eu +\frac{1}{n} \sum^{m+n}_{i=m+1} \big({\cal R}_i -\eE[{\cal R}_i]\big),
\end{equation}
with
	\begin{equation*}
	\left\{
	\begin{split} 
		{\cal D}_i & \equiv (1-\tau) \e1_{\varepsilon_i <0}-\tau \e1_{\varepsilon_i \geq 0}, \\
{\cal R}_i	 &  \equiv \rho_\tau(\varepsilon_i -l^{-1} \eX^\top_i \eu) -\rho_\tau(\varepsilon_i) -l^{-1} {\cal D}_i \eX^\top_i \eu,\\
\bfW_n& \equiv l^{-1} \sum^{m+n}_{i=m+1}	{\cal D}_i \eX_i^\top.
	\end{split}
	\right.
\end{equation*}
 For $\eE[    {\cal P}_n]$, using the identity:
 \[
 \rho_\tau(x-y) - \rho_\tau(x)=y(\e1_{x < 0} - \tau) +\int^y_0(\e1_{x <t} - \e1_{x<0})dt,
 \]
 and the fact that $\eE[\e1_{\varepsilon <0}]=F(0)=0$ in assumption (A1),  we have:
 \begin{align*}
 	\eE[	{\cal P}_n] &= \frac{1}{n} \eE \big[\sum^{m+n}_{i=m+1}  \big(\rho_\tau(\varepsilon_i -l^{-1} \eX^\top_i \eu) -\rho_\tau(\varepsilon_i)\big)\big] \\
 	&=\frac{1}{n} \sum^{m+n}_{i=m+1} \eE \bigg[\int_{0}^{l^{-1} \eX^\top_i \eu} \e1_{0< \varepsilon_i < t} dt\bigg] \\
 	 & =\frac{1}{n} \sum^{m+n}_{i=m+1} \int_{0}^{l^{-1} \eX^\top_i \eu} \big(F(t) -F(0)\big) dt.
 \end{align*}
 We will use  that, for $v \rightarrow 0$ we have $\int_{0}^{v} \big(F(t) -F(0)\big) dt=2^{-1}f(0) v^2+o(v^2)$. Therefore, taking into account assumption (A2), we have:
  \begin{align*}
 	 &
 	=\frac{f(0)}{2} l^{-2} \eu^\top \bigg(\frac{1}{n} \sum^{m+n}_{i=m+1} \eX_i \eX_i^\top\bigg) \eu \big(1+o(1)\big).
 \end{align*}
 On the other hand, using the independence of the $\varepsilon_i$ given by assumption (A1) and the fact that $|{\cal R}_i| \leq l |\eX_i^\top \eu |$ with probability 1, we have:
 \begin{align*}
 &	\eE \big[\sum^{m+n}_{i=m+1} \big({\cal R}_i -\eE[{\cal R}_i]\big)^2\big]\\
 &	=\sum^{m+n}_{i=m+1} \eE \big[{\cal R}_i -\eE[{\cal R}_i]\big]^2  \leq \sum^{m+n}_{i=m+1} \eE^2 [{\cal R}_i]\\
 &	 \leq  l^{-2} \sum^{m+n}_{i=m+1} \big|\eX_i^\top \eu\big|^2 \eE[\e1_{|\varepsilon_i | < l^{-1}|\eX_i^\top \eu|}] \\
 &	 = l^{-2} \sum^{m+n}_{i=m+1}\big|\eX_i^\top \eu\big|^2  \big(F(l^{-1}|\eX_i^\top \eu|) - F(-l^{-1}|\eX_i^\top \eu|)\big)\\
 &	=o\big(l^{-2} \sum^{m+n}_{i=m+1}\big|\eX_i^\top \eu\big|^2\big).
 \end{align*}
For the last relation, we used $F(l^{-1}|\eX_i^\top \eu|) - F(-l^{-1}|\eX_i^\top \eu|) \leq c l^{-1} |\eX_i^\top \eu| = o(1)$, by assumptions (A1) and (A3). Thus, the random process ${\cal P}_n$ from (\ref{eq4}) becomes:
\begin{equation}
\begin{split}
	{\cal P}_n =&\frac{f(0)}{2} l^{-2} \eu^\top \bigg(\frac{1}{n} \sum^{m+n}_{i=m+1} \eX_i \eX_i^\top \bigg) \eu  \big(1+o_\PP(1)\big)\\
	&+\frac{l^{-1}}{\sqrt{n}}  \sum^{m+n}_{i=m+1} {\cal D}_i \frac{\eX_i^\top \eu}{\sqrt{n}}.
\end{split}
	\label{eq5}
\end{equation}
However, since $E[D_i] = 0$ and using assumption (A2), we have by the CLT: 
\[
\frac{1}{\sqrt{n}}  \sum^{m+n}_{i=m+1} {\cal D}_i \eX_i  \overset{\cal L} {\underset{n \rightarrow \infty}{\longrightarrow}}  {\cal N}(\oo_p,\tau (1-\tau)\eU)\equiv \bfW.
\]
Thus, we can write for relation (\ref{eq5}):
\begin{equation}
\label{eq6}
{\cal P}_n=\bigg(\frac{f(0)}{2} l^{-2} \eu^\top \eU \eu +\frac{l^{-1}}{\sqrt{n}}\bfW \eu\bigg)  \big(1+o_\PP(1)\big).
\end{equation}
Relation (\ref{eq6}) obtained using assumptions (A1), (A2), (A3), will be used in the proof of Theorem \ref{TheoremTF 4.1}. The following theorem therefore provides the convergence rate and the asymptotic distribution of $\widehat\eb_n$ by considering several configurations of $m$, $n$, and $\lambda_n$. To achieve this, we consider the random process:
\begin{equation}
	\begin{split}
	Z_{m,n}(\eb) \equiv& \frac{1}{n}\sum_{i=m+1}^{m+n} \rho_\tau(Y_i- \eX_i^\top \eb)+ \frac{\lambda_n}{n} \sum_{j=1}^{p} v_{m,j} |\beta_j |\\
&+\frac{\eta_n}{n} \sum_{j=1}^{p} \omega_{m,j}|\beta_j - \widetilde \beta_{m,j}|.
\end{split}
\label{relZ}
\end{equation}

\begin{theorem}
\label{TheoremTF 4.1}
Under assumptions (A1), (A2), (A3) we have:\\
(i)  If $m^{\gamma_1/2}n^{-1/2} \lambda_n {\underset{(m,n) \rightarrow \infty}{\longrightarrow}} \infty$,  $n^{-1/2}\eta_n {\underset{n \rightarrow \infty}{\longrightarrow}} \infty$, $\eta_n\lambda_n^{-1} {\underset{n \rightarrow \infty}{\longrightarrow}} \infty$, $m^{-(\gamma_1+\gamma_2)/2}\eta_n\lambda_n^{-1}  {\underset{(m,n) \rightarrow \infty}{\longrightarrow}} 0$, and $m^{{\gamma_1}/2} \lambda_n\eta_n^{-1} {\underset{(m,n) \rightarrow \infty}{\longrightarrow}} r_0 \geq 0$, then 
\begin{equation}
	\label{eqTF30}
	m^{1/2} (\widehat\eb_n-\eb^*)    \overset{\cal L} {\underset{(m,n) \rightarrow \infty}{\longrightarrow}} 
	\left\{
	\begin{array}{ll}
		0, &\quad \text{if} \quad j \in {{\cal S}^*}^c,\\
		z_j,& \quad \text{if} \quad j \in {\cal S}^* ,
	\end{array}
	\right.
\end{equation} 
with $z_j \equiv \lim_{m \rightarrow \infty} z_{m,j}$, the asymptotic distribution of $z_{m,j}$. The distribution of $z_j$ is ${\cal N}(0, \tau (1-\tau) f^{-2}(0)\ \Upsilon^{-1}_{jj})$.\\
(ii)  If $m^{\gamma_1/2} n^{-1/2} \lambda_n  {\underset{(m,n) \rightarrow \infty}{\longrightarrow}}\infty$, $n^{-1/2} \lambda_n {\underset{n \rightarrow \infty}{\longrightarrow}} \lambda_0 \geq 0$ and $n^{-1/2} \eta_n {\underset{n \rightarrow \infty}{\longrightarrow}} \eta_0 \geq 0$, then
\begin{equation*}
	\label{eqTF37}
	\begin{split}
	n^{1/2} (\widehat\eb_n-\eb^*)    \overset{\cal L} {\underset{(m,n) \rightarrow \infty}{\longrightarrow}} \argmin_{\eu \in {\cal U}}\bigg(\frac{f(0)}{2}\eu^\top \eU \eu + \eu^\top \bfW \\
	+\sum_{j \in {\cal S}^*}\big(\lambda_0 \frac{\sgn(\beta^*_j)}{|\beta^*_j|^{\gamma_1}}u_j+\eta_0|\beta^*_j|^{\gamma_2}|u_j-r_0^{1/2}z_j|\big)\bigg),
\end{split}
\end{equation*} 
with ${\cal U} \equiv \{ \eu \in \R^p; \eu_{{{\cal S}^*}^c}=\oo\}$.\\
(iii) If $n^{-1/2} \lambda_n{\underset{n \rightarrow \infty}{\longrightarrow}} \infty$, $n^{-1} \lambda_n{\underset{n \rightarrow \infty}{\longrightarrow}} 0$, $\lambda_n  \eta_n^{-1} {\underset{n \rightarrow \infty}{\longrightarrow}} \infty$, then
\begin{equation*}
	\label{eqTF39}
	\begin{split}
	& \frac{n}{\lambda_n} (\widehat\eb_n-\eb^*)    \overset{\cal L} {\underset{n \rightarrow \infty}{\longrightarrow}}\\
&	  \argmin_{\eu \in {\cal U}}\bigg(\frac{f(0)}{2}\eu^\top \eU \eu +\sum_{j \in {\cal S}^*} \frac{\sgn(\beta^*_j)}{|\beta^*_j|^{\gamma_1}}u_j\bigg),
	 \end{split}
\end{equation*} 
with ${\cal U} \equiv \{ \eu \in \R^p ;  \eu_{{{\cal S}^*}^c}=\oo\}$.
\end{theorem}
Using this theorem, we can identify the configuration that will produce the most rapid convergence of the estimator $\widehat \eb_n$ to the true value. If $m \gg n$, the fastest convergence rate occurs in case {\it (i)}. In this case, the limiting distribution for non-zero estimators is the distribution of the estimator  $\widetilde \beta_{m,j}$ which is calculated from the source observations. Note that in case {\it (iii)}, the random variable $n\lambda_n^{-1} (\widehat \beta_{m,j} -\beta^*_j)$ could converge  to a nonzero constant for $j \in {\cal S}^*$. Therefore, the estimator is asymptotically biased. Furthermore, the convergence rate in case \textit{(iii)} is slower than that in case \textit{(ii)}. In order to obtain an estimator that converges rapidly, is asymptotically unbiased, and converges to the same distribution as the quantile estimator for nonzero coefficients, case {\it (i)} is recommended. In that case, we must choose $m \gg n$ so that, under case {\it (i)}, the estimator $\widehat \eb_n$ has the fastest convergence rate.\\		
The following theorem shows that the adaptive transfer LASSO quantile estimator  satisfies the sparsity property. To prove this, let the set 
$$\widehat {\cal S}_{m,n} \equiv \{j \in \{1, \ldots, p\};\widehat \beta_{n;j} \neq 0\},$$
which is an estimator for ${{\cal S}^*}$.
\begin{theorem}
	\label{TheoremTF_4.3}
Under assumptions (A1), (A2) together with $\max_{m+1 \leqslant i \leqslant m+n} \| \eX_i \|_2$  bounded, we have $\PP[\widehat{\cal S}_{m,n}={\cal S}^*] {\underset{(m,n) \rightarrow \infty}{\longrightarrow}}  1$.
\end{theorem}

Note that if we had directly applied the adaptive LASSO  quantile method to both the target and source data: $Y_i=\eX_i^\top \eb +\varepsilon_i$ for $i=1, \cdots, m+n$, in the case $m \geq n$, for $m^{-1/2}\lambda_n \rightarrow 0$, $m^{\gamma_1/2 - 1} \lambda_{m+n} \rightarrow \infty$, the estimator 
\begin{equation}
	\begin{split}
\widehat \eb_{m+n}^{aqLASSO} \equiv \argmin_{\eb \in \R^p} \bigg(\sum_{i=1}^{m+n} \rho_\tau(Y_i-\eX_i^\top \eb) \\
+\lambda_{m+n} \sum^p_{j=1} \widetilde v_{m+n,j}|\beta_j|\bigg),
\end{split}
\label{eq_aqLASSOm+n}
\end{equation}
 with  $(\widetilde v_{m+n,j})_{1 \leqslant j \leqslant p} \equiv |\argmin_{\eb \in \R^p}\sum_{i=1}^{m+n} \rho_\tau(Y_i-\eX_i^\top \eb)|^{-\gamma_1}$, would have a convergence rate of the order $(m+n)^{-1/2}$, that is, $m^{-1/2}$, the same as in case \textit{(i)} obtained in Theorem \ref{TheoremTF 4.1}. The adaptive LASSO quantile  estimator $\widehat \eb_{m+n}^{aqLASSO}$ satisfies the sparsity property and the asymptotic  normality: $m^{-1/2}(\widehat \eb_{m+n}^{aqLASSO} -\eb^*) \overset{\cal L} {\underset{m+n \rightarrow \infty}{\longrightarrow}}  {\cal N}(\oo_p, \tau (1-\tau) f^{-2}(0)\eU^{-1}_{{\cal S}^*})$, that is, the same law obtained in our case in Theorem \ref{TheoremTF 4.1}(i). For more details, see the article by \cite{Ciuperca-16}. In Section \ref{section_numeric}, we will compare the estimators $\widehat \eb_{m+n}^{aqLASSO}$ and $\widehat \eb_{n}$ through   simulations.\\
 
Note that  \cite{Huang.Wang.Wu.22} also propose and study a transfer  learning procedure within the framework of high-dimensional quantile regression models. However, the transfer estimator  differs from ours because  it uses an empirical quantile loss  only for the target data, using this to anchor an initial $L_1$-penalized quantile regression  estimator relative to the target by adjusting the contrast.   The asymptotic properties shown are the  convergence rate  and  the weak oracle property  of the variable selection procedure. On the other hand, compared to  the naive estimator with only the target data, the transfer learning LASSO quantile estimator defined by \cite{Jin.Yan.Aseltine.Chen.24}  achieves a much lower error  rate as a function of the sample sizes. However, the asymptotic distribution of the transfer learning LASSO estimators  studied in the papers \cite{Huang.Wang.Wu.22} and  \cite{Jin.Yan.Aseltine.Chen.24} is not examined. 
 
 \section{Algorithm}
\label{section_alg}
We now propose a transfer learning  algorithm based on the subgradient method to numerically compute $\widehat \eb_n$.\\
To derive the algorithm, we begin by reviewing the KKT conditions given by (\ref{eq7a}), (\ref{eq7b}), (\ref{eq7c}), which are found in  Section \ref{section_proofs} in the proof of Theorem \ref{TheoremTF_4.3}.
These three relations can generally be written as follows:
\begin{multline}
\label{ea1} 
\tau \sum^{m+n}_{i=m+1} X_{ji}-\sum^{m+n}_{i=m+1} X_{ji} \e1_{Y_i < \eX_i^\top   \eb}\\ -\lambda_n v_{m,j} s_{1j} -\eta_n \omega_{m,j} s_{2j}=0,
\end{multline}
with
\[
s_{1j} \equiv 	\left\{
\begin{array}{ll}
	 \frac{\beta_j}{|\beta_j|}, &\quad \text{if} \quad \beta_j \neq 0,\\
	\leq 1,& \quad \text{if} \quad \beta_j = 0 
\end{array}
\right.
\]
and
\[
s_{2j} \equiv 	\left\{
\begin{array}{ll}
	\frac{\beta_j -\widetilde \beta_j}{|\beta_j-\widetilde \beta_j|}, &\quad \text{if} \quad \beta_j \neq \widetilde \beta_{m,j},\\
	\leq 1,& \quad \text{if} \quad \beta_j = \widetilde \beta_{m,j} .
\end{array}
\right.
\]
Let $ \eX_{-j,i}$ be the vector $\eX_i$ without $X_{ji}$ and $\eb_{-j}$ be the vector $\eb$ without $\beta_j$. \\
If $\beta_j=0$, then relation (\ref{ea1}) becomes 
\begin{multline*}
	\label{ea1b}
	\tau \sum^{m+n}_{i=m+1} X_{ji}-\sum^{m+n}_{i=m+1} X_{ji} \e1_{Y_i <  \eX_{-j,i}^\top   \eb_{-j}} \\
	-\lambda_n v_{m,j} s_{1j} -\eta_n \omega_{m,j} s_{2j}=0,
\end{multline*}
from which we obtain
\begin{equation*}
	\begin{split}
		\label{ea2}
s_{1j}&=\frac{1}{\lambda_n v_{m,j}} \bigg( 	 \sum^{m+n}_{i=m+1} X_{ji}\big( \tau-  \e1_{Y_i <  \eX_{-j,i}^\top   \eb_{-j}}\big) \\ 
& \qquad -\eta_n \omega_{m,j} s_{2j}\bigg),
	\end{split}
\end{equation*}
with $s_{2j}=\sgn(\widetilde \beta_{m,j})$. Then, if $|s_{1j}|< 1$ then we take $\beta_j^{(k)}=0$ and 
\[
\begin{split}
&\frac{\beta_j}{|\beta_j|}=\\
&\frac{1}{\lambda_n v_{m,j}} \big( 	 \sum^{m+n}_{i=m+1} X_{ji}\big( \tau-  \e1_{Y_i <  \eX_{-j,i}^\top   \eb_{-j}}\big)  -\eta_n \omega_{m,j} s_{2j}\big).
\end{split}
\]
If $\beta_j \neq 0$ and  $\beta^{(k-1)}_j =\widetilde \beta_{m,j}$ then
\begin{equation*}\begin{split}
	&s_{2j}=\\
	&\frac{1}{\eta_n \omega_{m,j}} \big( 	 \sum^{m+n}_{i=m+1} X_{ji}\big( \tau-  \e1_{Y_i <  \eX_{-j,i}^\top   \eb_{-j}}\big)  -\lambda_n v_{m,j} s_{1j}\big).
\end{split}
\end{equation*}
If $|s_{2j}| <1$ then $\beta^{(k)}=\widetilde \beta_{m,j}$.\\
If $\beta_j \neq 0$ then $|s_{1j}| \geq 1$. If $\beta^{(k-1)}_j \neq \widetilde  \beta_j$ then
\begin{align*}
\frac{\beta_j}{|\beta_j|}=&\frac{1}{\lambda_n v_{m,j}} \bigg( 	 \sum^{m+n}_{i=m+1} X_{ji}\big( \tau-  \e1_{Y_i <  \eX_{-j,i}^\top   \eb_{-j}}\big) \\& \qquad  -\eta_n \omega_{m,j}  \frac{\beta_j-\widetilde \beta_{m,j}}{|\beta_j-\widetilde \beta_{m,j}|}\bigg),
\end{align*}
from which we obtain
\[
\begin{split}
&	\eta_n \omega_{m,j}  \frac{\beta_j-\widetilde \beta_{m,j}}{|\beta_j-\widetilde \beta_{m,j}|}\\
& =  \sum^{m+n}_{i=m+1} X_{ji}\big( \tau-  \e1_{Y_i <  \eX_{-j,i}^\top   \eb_{-j}}\big) -\frac{n \lambda_n \beta_j}{|\beta_j|}.
\end{split}
\]
From this last equation, we can derive $\beta_j$:
	\begin{align*}
&\sum^{m+n}_{i=m+1} X_{ji}\big( \tau-  \e1_{Y_i <  \eX_{-j,i}^\top   \eb_{-j}}\big)=\lambda_n v_{m,j} \frac{\beta_j}{|\beta_j|}\\
&\qquad +\eta_n \omega_{m,j} \frac{\beta_j -\widetilde \beta_{m,j}}{ |\beta_j -\widetilde \beta_{m,j}|},
\end{align*}
from which we can write:
\begin{align*}
&|\beta_j|  |\beta_j -\widetilde \beta_{m,j}| \sum^{m+n}_{i=m+1} X_{ji}\big( \tau-  \e1_{Y_i <  \eX_{-j,i}^\top   \eb_{-j}}\big)\\
&=\lambda_n  v_{m,j} \beta_j |\beta_j -\widetilde \beta_{m,j}|+\eta_n\omega_{m,j} \big(\beta_j - \widetilde \beta_{m,j}\big)|\beta_j|,
\end{align*}
that is to say
\begin{align*}
&\beta_j  \big( \lambda_n v_{m,j} |\beta_j -\widetilde \beta_{m,j}| +\eta_n \omega_{m,j} |\beta_j|\big)\\
&=|\beta_j|  |\beta_j -\widetilde \beta_{m,j}|  \sum^{m+n}_{i=m+1} X_{ji}\big( \tau-  \e1_{Y_i <  \eX_{-j,i}^\top   \eb_{-j}}\big)\\ &\qquad +\eta_n \omega_{m,j} \widetilde \beta_{m,j} |\beta_j|,
\end{align*}
which implies:
\begin{equation}
	\label{ea3}
	\begin{split}
		&		\beta_j=\
		\big(|\beta_j|  |\beta_j -\widetilde \beta_{m,j}| \sum^{m+n}_{i=m+1} X_{ji}\big( \tau-  \e1_{Y_i <  \eX_{-j,i}^\top   \eb_{-j}}\big) \\
		&+\eta_n \omega_{m,j} \widetilde \beta_{m,j} |\beta_j|\big)/\big(\lambda_n v_{m,j} |\beta_j -\widetilde \beta_{m,j}| +\eta_n \omega_{m,j} |\beta_j|\big).
	\end{split}
\end{equation}
After examining all of these situations, we are now able to propose the following algorithm:
\begin{algorithm}[H]
\algorithmicrequire{ The  source data $(Y_i,\eX_i)_{1 \leqslant i \leqslant m}$and the target data 
		 $(Y_i,\eX_i)_{m+1 \leqslant i \leqslant m+n}$.}\\
\algorithmicensure{ Calculate $\widehat \eb_{n}$ using the following steps.}
	\begin{enumerate}
		\item  The estimator  $\widetilde \eb_m$ is calculate on the source data.
		\item\textit{Step 0}: We start with a $\eb^{(0)}$. 
		\item \textit{Step $k$},  calculates the estimate $\eb^{(k)}$ given the value $\eb^{(k-1)}$ of step \textit{(k-1)}. For $j$ from 1 to $p$, we perform the following:
		\begin{enumerate}
			\item Calculate $s_{1j}$. If $|s_{1j}|<1$ then we take $\beta^{(k)}_j=0$. In other words, if
			\begin{align*}
				\frac{1}{\lambda_n v_{m,j}} \bigg|	 \sum^{m+n}_{i=m+1} X_{ji}\big( \tau-  \e1_{Y_i <  \eX_{-j,i}^\top   \eb_{-j}^{(k-1)}}\big) \\
				-\eta_n \omega_{m,j} \sgn(\beta_j^{(k-1)}- \widetilde \beta_{m,j})\bigg|<1,
			\end{align*}
			then $\beta^{(k)}_j=0$.\\
			If $s_{1j} \geq 1$  then we execute (b) when $|s_{2j}|<1$ or (c) when $|s_{2j}|\geq 1$.
			\item Calculate $s_{2j}$. If $|s_{2j}|<1$ then we take $\beta^{(k)}_j=\widetilde \beta_{m,j}$. In other words, if 
			\begin{align*}
				\frac{1}{\eta_n \omega_{m,j}} \bigg|	 \sum^{m+n}_{i=m+1} X_{ji}\big( \tau-  \e1_{Y_i <  \eX_{-j,i}^\top   \eb_{-j}^{(k-1)}}\big) \\
				-\lambda_n v_{m,j} \sgn(\beta_{m,j}^{(k-1)})\bigg|<1,
			\end{align*}
			then $\beta^{(k)}_j=\widetilde \beta_{m,j}$.
			\item Otherwise, calculate $|\beta_j|$ using equation (\ref{ea3}). 
			We will replace only $|\beta_j|$ with\\
			$
			b_j \equiv \big(\big| |\beta_j^{(k-1)}|   |\beta_j^{(k-1)} -\widetilde \beta_{m,j}| \sum^{m+n}_{i=m+1} X_{ji}\big( \tau-  \e1_{Y_i <  \eX_{-j,i}^\top   \eb_{-j}^{(k-1)}}\big) +\eta_n \omega_{m,j} \widetilde \beta_{m,j} |\beta_j^{(k-1)}|\big|\big)/
			\big(\lambda_n v_{m,j} |\beta_j^{(k-1)} -\widetilde \beta_{m,j}| +\eta_n \omega_{m,j} |\beta_j^{(k-1)}|\big).	$\\
			Hence\\
			$
			\beta_j^{(k)}=\big( b_j  |\beta_j^{(k-1)} -\widetilde \beta_{m,j}| \sum^{m+n}_{i=m+1} X_{ji}\big( \tau-  \e1_{Y_i <  \eX_{-j,i}^\top   \eb_{-j}^{(k-1)}}\big) +\eta_n \omega_{m,j} \widetilde \beta_{m,j} b_j\big)/ \big(\lambda_n v_{m,j} |\beta_j^{(k-1)} -\widetilde \beta_{m,j}| +\eta_n \omega_{m,j} b_j\big)$.
		\end{enumerate}
		\item \textit{Stop}: when $\| \eb^{(k)} -\eb^{(k-1)}\|_2 \leq \epsilon$, with $\epsilon$ a threshold set in advance.	
	\end{enumerate}
		\caption{}
\end{algorithm}
We note that, based on the results of \cite{Tseng.01}, the algorithm is convergent.\\

 An article that proposes an algorithm for calculating a transfer LASSO estimator is the one by \cite{Takada.Fujisawa.20}, but it is  completely different  from ours because the regression method used was ordinary least squares regression. Again, for a loss differentiable function, \cite{Tian.Feng.23} introduce a general transfer learning algorithm on GLMs to transfer all sources in a given index set, while \cite{Liu.Song.25} propose a  two-step transfer algorithm based on expectile regression model.  For high-dimensional quantile regression, \cite{Qiao-He-Zhou.24} propose transfer learning algorithms in which, to counter non-smoothness and non convexity of the quantile loss, they employ the convolution-type smoothed quantile regression. \cite{Zhang.Zhu.25} also proposes a  transfer learning algorithm based on convolution smoothing of the regression quantile.
 \section{Numerical study}
 \label{section_numeric}
 In this section, we will conduct a numerical study of $\widehat \eb_{n}$ defined by (\ref{eq_cible}) and calculated using the algorithm proposed in Section \ref{section_alg}.  Unless otherwise noted, the results presented here are obtained from 200 Monte Carlo replications.  For the algorithm, the value of $\epsilon$ is $10^{-3}$.  \\
For  $i=1,\cdots, m, m+1, \cdots, m+n$, we consider the design  $X_{ji} \sim {\cal N}(0,1)$ for $j \in \{1, \cdots, p\}\setminus \{3, 4, 5, 8\}$, $X_{3i}\sim {\cal N}(2,1)$, $X_{4i}\sim {\cal N}(4,1)$, $X_{5i}\sim {\cal N}(1,1)$, $X_{8i}\sim {\cal N}(1,2)$ and are iid. In most of the following simulations, the errors are either ${\cal N}(0,1)$ or $0.2 {\cal N}(0,1) +{\cal N}^2(0,1)$. For each distribution of $\varepsilon$, we calculate the quantile index estimation using $\widehat \tau_n =(m+n)^{-1} \sum^{m+n}_{i=1} \varepsilon_i \e1_{\varepsilon_i <0}$.\\
The values obtained for  $\widehat \tau_n$ are  0.50 and 0.14 when $ \varepsilon \sim {\cal N}(0,1)$ or $0.2 {\cal N}(0,1) +{\cal N}^2(0,1)$,  respectively.  The initial value $\eb^{(0)}$ taken for the algorithm is a $p$-dimensional vector with all components equal to 1, but you can use any other value; the algorithm will always converge to same result.
The simulations,  performed on a computer Inter(R), Core (CM), 1.6 Ghz, were performed using the R software, with quantile estimates calculated using the \textit{rq} function from the \textit{quantreg} package. The adaptive transfer LASSO quantile estimation  was implemented by applying the algorithm presented in Section \ref{section_alg}.
 \subsection{Comparison of $\widehat \eb_n$ with adaptive LASSO quantile  estimators}
 \label{subsec_acomp_avec_LASSO}
  The adaptive transfer LASSO quantile estimator $\widehat \eb_{n}$ is compared with two adaptive  LASSO quantile estimators calculated as follows:
 \begin{itemize}
 	\item  on both the target and source data as the minimizer of the random process:
 	\begin{multline*}
 	\sum^{m+n}_{i=1} \rho_\tau(Y_i-\eX_i^\top \eb)\\+
 	(m+n)^{2/5} \sum^p_{j=1}   | \widetilde\beta_{m+n,j}|^{-1.225} |\beta_j| ,
 	\end{multline*}
 	with $\widetilde \eb_{m+n}\equiv \argmin_{\eb \in \R^p}\sum^{m+n}_{i=1} \rho_\tau(Y_i-\eX_i^\top \eb) $. This estimator is denoted by $\widehat \eb^{aqLASSO}_{m+n}$ and corresponds to  estimator (\ref{eq_aqLASSOm+n}).
 	\item  only on the target data, such as the minimizer of the random process:
 	\[
 	\sum^{m+n}_{i=m+1} \rho_\tau(Y_i-\eX_i^\top \eb)+
 	n^{2/5} \sum^p_{j=1}   | \widetilde\beta_{n,j}|^{-1.225} |\beta_j| ,
 	\]
 	with $\widetilde \eb_{n}\equiv \argmin_{\eb \in \R^p}\sum^{m+n}_{i=m+1} \rho_\tau(Y_i-\eX_i^\top \eb) $. This estimator is denoted $\widehat \eb^{aqLASSO}_n$.
 \end{itemize} 
 These two adaptive  LASSO quantile estimations were calculated using the \textit{rq} function from the R package \textit{quantreg}.\\
Unless otherwise stated, for all simulations presented in this subsection, we assume $m=n^2$ and are under the conditions {\it (i)} of Theorem \ref{TheoremTF 4.1}. Thus, for the estimate $\widehat \eb_{n}$ of equation (\ref{eq1_cible}), we take: $\gamma_1=2$ and $\gamma_2=1.25$, $\lambda_n=n^{1/2}$ and $\eta_n=n^{1/2+max(\gamma_1+0.001,(\gamma_1+\gamma_2)/2)}$. The number  $p$ of explanatory variables will vary, but we will always keep three non-zero coefficients: $|{\cal S}^*|=3$, more precisely ${\cal S}^*=\{1, 2, 3\}$.
For 200 Monte Carlo simulations, in Tables \ref{Tabl1}, \ref{Tabl2}, \ref{Tabl3} and in Figures \ref{fig_betapetit_p10}, \ref{fig_betapetit_p100}, and \ref{fig_betagrand_p10}, we have plotted:
\begin{itemize}
	\item  the {\it estimation error} calculated for each of the three estimators as  $\frac{\| (\widehat \eb -\eb^*)_{{\cal S}^*} \|_2}{ \|   \eb^* \|_2}$, 
	\item  the {\it false zero rate} by $\frac{| \{ j \in {\cal S}^*; \widehat \beta_j = 0\}|}{ |{\cal S}^*|}$ 
	\item and the {\it false non-zero rate} calculated by  $\frac{| \{ j \in {{\cal S}^*}^c; \widehat \beta_j \neq 0\}|}{ |{{\cal S}^*}^c|}$. 
\end{itemize}
For Figures \ref{fig_betapetit_p10}, \ref{fig_betapetit_p100}, and \ref{fig_betagrand_p10}, the model errors follow a distribution of the form $0.2 {\cal N}(0,1) +{\cal N}^2(0,1)$. In the case where all nonzero coefficients are close to zero, if $(\varepsilon_i)_{1 \leqslant i \leqslant m+n} \sim{\cal N}(0,1)$,  based on Table \ref{Tabl1}, we deduce that  the estimators $\widehat \eb_{n}$, $\widehat \eb^{aqLASSO}_{m+n}$ have roughly the same accuracy and level of selection of significant and non-significant variables. If $ \varepsilon \sim 0.2 {\cal N}(0,1) +{\cal N}^2(0,1)$, i.e., a asymmetric distribution, then $\widehat \eb_n$ is less biased than $\widehat \eb^{aqLASSO}_{m+n}$ regardless of whether $p$ is small or large, and whether it is of the same order of magnitude as $n$ or not (see Figures \ref{fig_betapetit_p10} and \ref{fig_betapetit_p100}). On the other hand, based on Table \ref{Tabl2} and Figure \ref{fig_betagrand_p10}, we can conclude that if the nonzero parameters are further from 0, then these two estimators are similar regardless of whether the errors follow  ${\cal N}(0,1)$ or $0.2 {\cal N}(0,1) +{\cal N}^2(0,1)$, except for $p=20$, $n=50$, and $\varepsilon \sim {\cal N}(0,1)$. The estimations $\widehat \eb^{aqLASSO}_{n}$, since they are calculated on a smaller number of observations than $\widehat \eb^{aqLASSO}_{m+n}$, may perform less well for $\varepsilon \sim {\cal N}(0,1)$ when $n$ and $p$ are of the same order of magnitude.
\begin{table*} 
		\centering 
					{\scriptsize
		\begin{tabular}{cccccc}\hline 
			 $p$ &$n$   &estimator &  $\frac{\| (\widehat \eb -\eb^*)_{{\cal S}^*} \|_2}{ \|   \eb^* \|_2}$  &    $\frac{| \{ j \in {\cal S}^*; \widehat \beta_j = 0\}|}{ |{\cal S}^*|}$  &   $\frac{| \{ j \in {{\cal S}^*}^c; \widehat \beta_j \neq 0\}|}{ |{{\cal S}^*}^c|}$\\ [4pt]   \hline \toprule
				5	  & 20   & $\widehat \eb_{n}$  &  0.47 & 0.33 &0.17  \\[4pt]
			 &  & $\widehat \eb^{aqLASSO}_{m+n}$ & 0.45&0.44 & 0.04\\[4pt]
			&  & $\widehat \eb^{aqLASSO}_{n}$ &0.94 & 0.56 & 0.20 \\[4pt] \cmidrule{2-6} 
		& 40  & $\widehat \eb_{n}$  &  0.24 & 0.07 &0.09 \\[4pt]
	& &  $\widehat \eb^{aqLASSO}_{m+n}$ & 0.30& 0.05 & 0.002\\[4pt]
	& &  $\widehat \eb^{aqLASSO}_{n}$ &0.76 & 0.59 & 0.11 \\ [4pt] \cmidrule{2-6}
	  & 100  & $\widehat \eb_{n}$  &  0.08 & 0.01 &0.04 \\[4pt]
	& &  $\widehat \eb^{aqLASSO}_{m+n}$ & 0.09 &0 & 0\\[4pt]
	& &  $\widehat \eb^{aqLASSO}_{n}$ &0.46 & 0.53 & 0.06 \\ [4pt]  \cmidrule{2-6}
	& 1000  & $\widehat \eb_{n}$  &  0.006 & 0 &0 \\[4pt]
	& &  $\widehat \eb^{aqLASSO}_{m+n}$ & 0.006&0 & 0\\ [4pt]
	& & $\widehat \eb^{aqLASSO}_{n}$ &0.35 & 0.18 & 0.01 \\[4pt] \cmidrule{1-6}
	25	& 50 & $\widehat \eb_{n}$  &  0.43 & 0.29 &0.09 \\ [4pt]
	& &  $\widehat \eb^{aqLASSO}_{m+n}$ & 0.25&0 & 0\\ [4pt]
	& &  $\widehat \eb^{aqLASSO}_{n}$ &0.81 & 0.54 & 0.03 \\ [4pt] \cmidrule{2-6}
 & 250 & $\widehat \eb_{n}$  &  0.09 & 0.04 &0.02\\[4pt]
	& &  $\widehat \eb^{aqLASSO}_{m+n}$ & 0.03&0 & 0\\ [4pt]
	& &  $\widehat \eb^{aqLASSO}_{n}$ &0.50 & 0.49 & 0.001 \\ [4pt]\hline
						\end{tabular}
}
	\caption{ Comparison study of $\widehat \eb_{n}$, $\widehat \eb^{aqLASSO}_{m+n}$, $\widehat \eb^{aqLASSO}_{n}$, when   $\beta_1^*=0.1$, $\beta^*_2=0.2$, $\beta^*_3=-0.1$, $\varepsilon \sim {\cal N}(0,1)$.  } 
\label{Tabl1} 
\end{table*}
\begin{table*} 
\begin{center}
				{\scriptsize
		\begin{tabular}{BBBBBBB}\hline
			 $p$& $n$ & $\varepsilon$  & estimator &  $\frac{\| (\widehat \eb -\eb^*)_{{\cal S}^*} \|_2}{ \|   \eb^* \|_2}$  &    $\frac{| \{ j \in {\cal S}^*; \widehat \beta_j = 0\}|}{ |{\cal S}^*|}$  &  $\frac{| \{ j \in {{\cal S}^*}^c; \widehat \beta_j \neq 0\}|}{ |{{\cal S}^*}^c|}$\\ [4pt] \hline  \toprule
			  & 20& ${\cal N}(0,1)$ & $\widehat \eb_{n}$  &  0.11 & 0.08 &0.12 \\[4pt]
			& & & $\widehat \eb^{aqLASSO}_{m+n}$& 0.03&0  & 0\\[4pt]
			& & & $\widehat \eb^{aqLASSO}_{n}$ &0.18 & 0.0006 & 0.13 \\[4pt] \cmidrule{2-7}
		 5 & 100 & ${\cal N}(0,1)$ & $\widehat \eb_{n}$  &  0.007 & 0 &0.02 \\[4pt]
			& & & $\widehat \eb^{aqLASSO}_{m+n}$& 0.007&0 & 0\\[4pt]
			& & & $\widehat \eb^{aqLASSO}_{n}$ &0.07 & 0 & 0.03 \\[4pt] \cmidrule{1-7}
		 & 50 & ${\cal N}(0,1)$ & $\widehat \eb_{n}$  &  0.50 & 0.58 &0.09 \\[4pt]
			& & & $\widehat \eb^{aqLASSO}_{m+n}$& 0.01 &0 & 0\\[4pt]
			& & & $\widehat \eb^{aqLASSO}_{n}$ &0.09 & 0. & 0.02 \\[4pt] \cmidrule{2-7}
		25	 & 250 & ${\cal N}(0,1)$ & $\widehat \eb_{n}$  &  0.05 & 0.05  &0.02\\[4pt]
			& & & $\widehat \eb^{aqLASSO}_{m+n}$ & 0.002&0 & 0\\[4pt]
			& & & $\widehat \eb^{aqLASSO}_{n}$ &0.04 & 0 & 0  \\[4pt]   \hline \hline
				 & 20 & $0.2 {\cal N}(0,1) +{\cal N}^2(0,1)$ & $\widehat \eb_{n}$  &  0.01 & 0  &0  \\[4pt]
			& & & $\widehat \eb^{aqLASSO}_{m+n}$ & 0.01&0  & 0\\[4pt]
			& & & $\widehat \eb^{aqLASSO}_{n}$ &0.07 & 0  & 0.01 \\ [4pt]\cmidrule{2-7}
			5 & 100 & $0.2 {\cal N}(0,1) +{\cal N}^2(0,1)$ &$\widehat \eb_{n}$ &  0.002 & 0 &0 \\[4pt]
			& & & $\widehat \eb^{aqLASSO}_{m+n}$ & 0.002&0 & 0\\[4pt]
			& & & $\widehat \eb^{aqLASSO}_{n}$&0.02 & 0 & 0 \\ [4pt]\hline
			 & 50 & $0.2 {\cal N}(0,1) +{\cal N}^2(0,1)$ & $\widehat \eb_{n}$  &  0.006 & 0 &0 \\[4pt]
			& & & $\widehat \eb^{aqLASSO}_{m+n}$ & 0.005&0 & 0\\[4pt]
			& & & $\widehat \eb^{aqLASSO}_{n}$ &0.04 & 0. & 0.005 \\[4pt]\cmidrule{2-7}
			25 & 250 & $0.2 {\cal N}(0,1) +{\cal N}^2(0,1)$ & $\widehat \eb_{n}$  &  0.001 & 0  &0\\[4pt]
			& & & $\widehat \eb^{aqLASSO}_{m+n}$ & 0.0009&0 & 0\\[4pt]
			& & & $\widehat \eb^{aqLASSO}_{n}$ &0.01 & 0 & 0 \\ [4pt]\hline 
					\end{tabular}
	}
\end{center}
	\caption{Comparison study of $\widehat \eb_{n}$, $\widehat \eb^{aqLASSO}_{m+n}$, $\widehat \eb^{aqLASSO}_{n}$, when   $\beta_1^*=1$, $\beta^*_2=2$, $\beta^*_3=-1$,  $m=n^2$. } 
\label{Tabl2} 
\end{table*}
Table \ref{Tabl3} highlights another advantage of the transfer learning method: the runtime required to estimate the parameters. This advantage is particularly useful when $p$ is large, i.e., for high-dimensional models.
\begin{table*} 
		\begin{center}
				{\scriptsize
			\begin{tabular}{cccc}\hline
			 $p$ & $n$ & method &  time (sec)\\[4pt] \hline \toprule  
		10 & 50 & $\widehat \eb_{n}$  &  0.0169  \\[4pt]
		&  & $\widehat \eb^{aqLASSO}_{m+n}$ &  0.034 \\[4pt]
		& & $\widehat \eb^{aqLASSO}_{n}$ & 0.003  \\[4pt] \cmidrule{2-4}	
		 & 250 & $\widehat \eb_{n}$   &  2.64  \\[4pt]
		&  & $\widehat \eb^{aqLASSO}_{m+n}$ &  3 \\[4pt]
		& & $\widehat \eb^{aqLASSO}_{n}$ & 0.006  \\ [4pt]\cmidrule{1-4}		
			25 & 250 &$\widehat \eb_{n}$   &  3.26  \\[4pt]
		&  & $\widehat \eb^{aqLASSO}_{m+n}$&  5.70 \\[4pt]
		& & $\widehat \eb^{aqLASSO}_{n}$ & 0.007  \\[4pt] \cmidrule{1-4}
			200 & 250 & $\widehat \eb_{n}$   &  95  \\[4pt]
		&  & $\widehat \eb^{aqLASSO}_{m+n}$ &  150 \\[4pt]
		& & $\widehat \eb^{aqLASSO}_{n}$ & 0.067  \\[4pt] \hline
									\end{tabular}
								}
		\end{center}
		\caption{ Execution time for the each estimation,  when  $\beta_1^*=1$, $\beta^*_2=2$, $\beta^*_3=-1$,  $\varepsilon \sim 0.2 {\cal N}(0,1) +{\cal N}^2(0,1)$. } 
	\label{Tabl3} 
\end{table*}

 \begin{figure*}[h!] 
 	\begin{tabular}{ccc}
 		\includegraphics[width=0.3\linewidth,height=4.5cm]{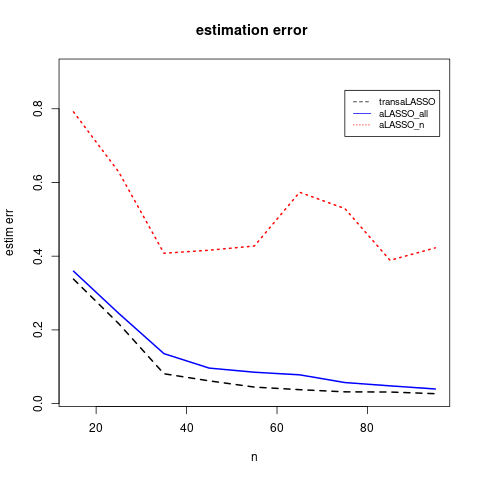} &
 		\includegraphics[width=0.3\linewidth,height=4.5cm]{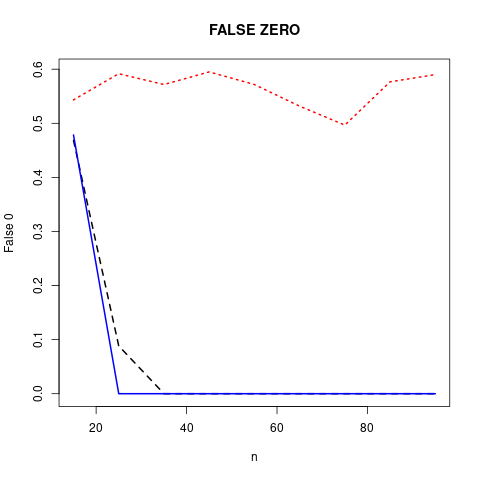}&
 		\includegraphics[width=0.3\linewidth,height=4.5cm]{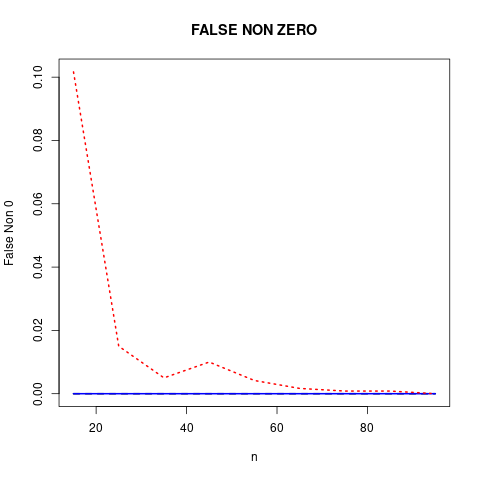} \\
 		{\small (a) Evolution of estimation error.} &
 		{\small (b)  Evolution of false  zeros.} &
 		{\small (a) Evolution of false non-zeros.}
 	\end{tabular}
 	\caption{\scriptsize  Evolution with $n$ of estimations $\widehat \eb_{n}$  (\textit{transaLASSO}, black dashed), $\widehat \eb^{aqLASSO}_{m+n}$ (\textit{aLASSO\_all}, blue line), $\widehat \eb^{aqLASSO}_{n}$ (\textit{aLASSO\_n}, red dotted line),   when  $\beta_1^*=0.1$, $\beta^2_0=0.2$, $\beta^*_3=-0.1$,  $p=10$, $\varepsilon \sim 0.2 {\cal N}(0,1) +{\cal N}^2(0,1)$. 	}
 	\label{fig_betapetit_p10}
 \end{figure*}
 
\begin{figure*}[h!] 
	\begin{tabular}{ccc}
		\includegraphics[width=0.3\linewidth,height=4.5cm]{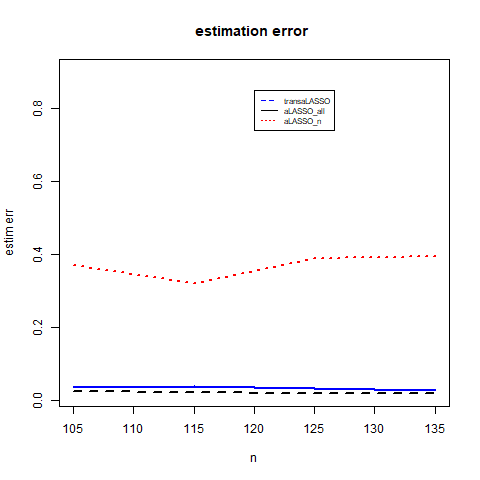} &
		\includegraphics[width=0.3\linewidth,height=4.5cm]{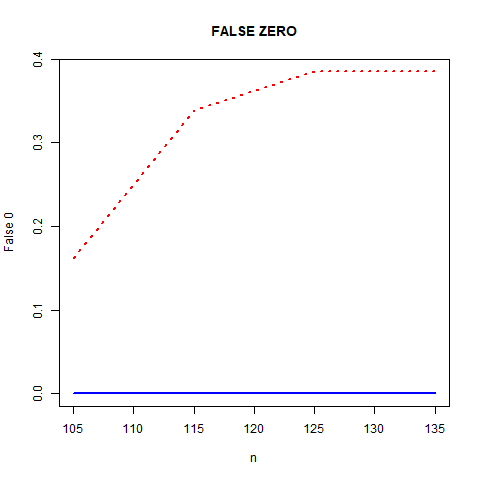}&
		\includegraphics[width=0.3\linewidth,height=4.5cm]{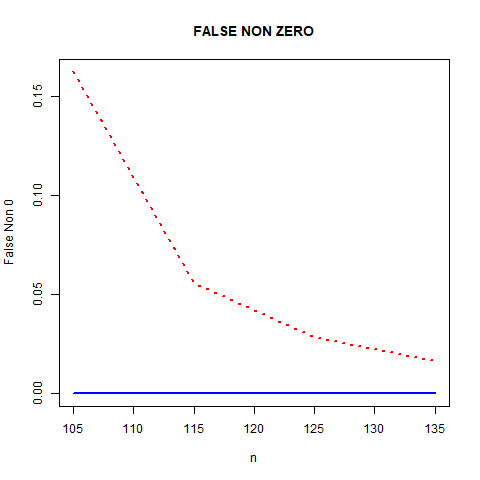} \\
		{\small (a) Evolution of estimation error.} &
		{\small (b)  Evolution of false  zeros.} &
		{\small (a) Evolution of false non-zeros.}
	\end{tabular}
	\caption{\scriptsize   Evolution   with $n$ of estimations $\widehat \eb_{n}$  (\textit{transaLASSO}, black dashed), $\widehat \eb^{aqLASSO}_{m+n}$ (\textit{aLASSO\_all}, blue line), $\widehat \eb^{aqLASSO}_{n}$ (\textit{aLASSO\_n}, red dotted line), when  $\beta_1^*=0.1$, $\beta^*_2=0.2$, $\beta^*_3=-0.1$,  $p=100$, $\varepsilon \sim 0.2 {\cal N}(0,1) +{\cal N}^2(0,1)$. 	}
	\label{fig_betapetit_p100}
\end{figure*}
\begin{figure*}[h!] 
	\begin{tabular}{ccc}
		\includegraphics[width=0.3\linewidth,height=4.5cm]{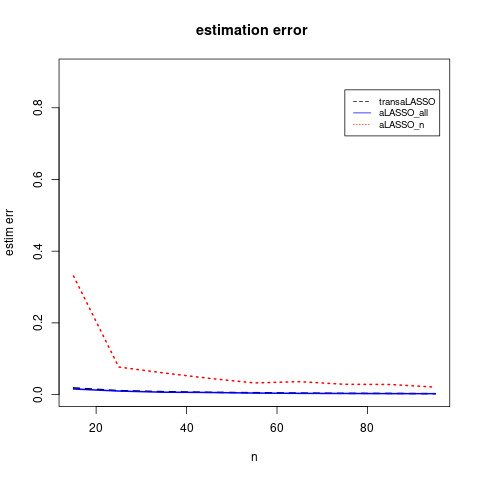} &
		\includegraphics[width=0.3\linewidth,height=4.5cm]{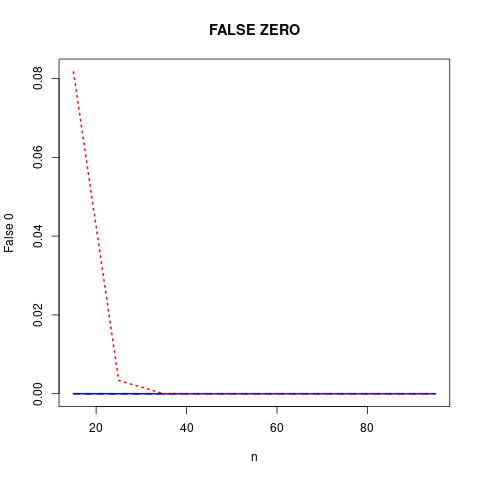}&
		\includegraphics[width=0.3\linewidth,height=4.5cm]{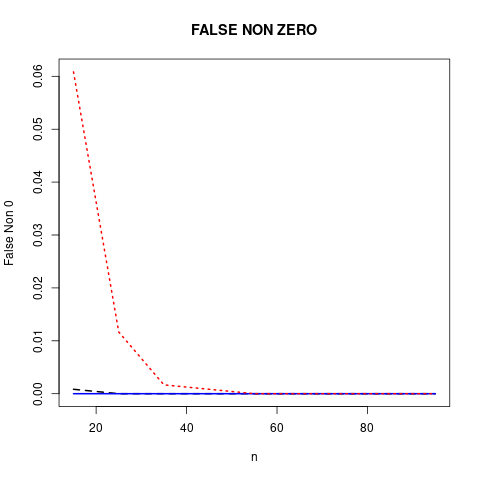} \\
		{\small (a) Evolution of estimation error.} &
		{\small (b)  Evolution of false  zeros.} &
		{\small (a) Evolution of false non-zeros.}
	\end{tabular}
	\caption{\scriptsize   Evolution   with $n$  of estimations $\widehat \eb_{n}$  (\textit{transaLASSO}, black dashed), $\widehat \eb^{aqLASSO}_{m+n}$ (\textit{aLASSO\_all}, blue line), $\widehat \eb^{aqLASSO}_{n}$ (\textit{aLASSO\_n}, red dotted line), when $\beta_1^*=1$, $\beta^*_2=2$, $\beta^*_3=1$,  $p=10$, $\varepsilon \sim 0.2 {\cal N}(0,1) +{\cal N}^2(0,1)$. 	}
	\label{fig_betagrand_p10}
\end{figure*}
Taking into account the results in Table \ref{Tabl2}, we examined in greater detail the comparison between $\widehat\eb_{n}$ and $\widehat \eb^{aqLASSO}_{m+n}$ when the values of $\beta^*_j$, $j \in {\cal S}^*$, are sufficiently far from 0 and $\varepsilon$ follow a ${\cal N}(0,1)$ distribution.  We either vary $m$ relative to $n$ or choose a different distribution for $\varepsilon$. In Table \ref{Tabl5}, for a fixed $n$, we take $m \in \{n^{5/2}, n^3, n^{7/2}\}$. The results will also be compared with two entries from Table \ref{Tabl2}. We conclude that if $p$ and $n$ are close and $\varepsilon \sim {\cal N}(0,1)$, then starting from $m \geq n^{5/2}$ the estimators $\widehat \eb_n$ and $\widehat \eb^{aqLASSO}_{m+n}$ yield the same results; that is, the source dataset must contain more than $m^2$ observations for the estimates obtained to be the same. Note that this holds only for the normal distribution; for the skewed distribution $ \sim 0.2 {\cal N}(0,1) +{\cal N}^2(0,1)$, both estimation methods yield the same results (see Table \ref{Tabl2}). The limited memory of the computer used, as well as the long computation time required to calculate $\widehat \eb^{aqLASSO}_{m+n}$, forced us to set the number of Monte Carlo replicates used for the results presented in Table \ref{Tabl5} to 50 when $m \in \{n^3, n^{7/2}\}$ and $p \in \{25, 100\}$. 
In Figure \ref{fig_Ubetapetit_p10}, for $m=n^2$, the error distribution is a uniform distribution ${\cal U}[-0.5,0.75]$ for which $\widehat \tau_n=0.40$. It follows that these two estimators behave similarly when the error distribution is uniform.\\
\begin{table*}[h] 
\begin{center}
	{\scriptsize
			\begin{tabular}{BBBBBBB}\hline  \toprule
					$p$& $n$ & $m$  & estimator &  $\frac{\| (\widehat \eb -\eb^*)_{{\cal S}^*} \|_2}{ \|   \eb^* \|_2}$  &    $\frac{| \{ j \in {\cal S}^*; \widehat \beta_j = 0\}|}{ |{\cal S}^*|}$  &  $\frac{| \{ j \in {{\cal S}^*}^c; \widehat \beta_j \neq 0\}|}{ |{{\cal S}^*}^c|}$\\   \hline  \toprule
				3	& 10& $n^{2}$ & $\widehat \eb_{n}$  &  0.30 & 0.20 &0.47 \\ 
				~	& & & $\widehat \eb^{aqLASSO}_{m+n}$& 0.10&0  & 0.01 \\  \cmidrule{3-7}
						& & $n^{5/2}$ & $\widehat \eb_{n}$  &  0.08 & 0.03 &0.05 \\ 
					& & & $\widehat \eb^{aqLASSO}_{m+n}$& 0.06&0  & 0.005 \\  \cmidrule{3-7}
					& & $n^{3}$ & $\widehat \eb_{n}$  &  0.03 & 0 &0 \\ 
					& & & $\widehat \eb^{aqLASSO}_{m+n}$& 0.03&0  & 0 \\  \cmidrule{3-7}
					& & $n^{7/2}$ & $\widehat \eb_{n}$  &  0.01 & 0 &0 \\ 
					& & & $\widehat \eb^{aqLASSO}_{m+n}$& 0.01&0  & 0 \\  \midrule 
				5	& 20& $n^{5/2}$ & $\widehat \eb_{n}$  &  0.03 & 0.01 &0.02 \\ 
					& & & $\widehat \eb^{aqLASSO}_{m+n}$& 0.02&0  & 0 \\  \cmidrule{3-7}
						& & $n^{3}$ & $\widehat \eb_{n}$  &  0.007 & 0 &0 \\ 
					& & & $\widehat \eb^{aqLASSO}_{m+n}$& 0.007&0  & 0 \\  \cmidrule{3-7}
						& & $n^{7/2}$ & $\widehat \eb_{n}$  &  0.004 & 0 &0 \\ 
					& & & $\widehat \eb^{aqLASSO}_{m+n}$& 0.004&0  & 0 \\  \cmidrule{1-7}
					25	& 50& $n^{5/2}$ & $\widehat \eb_{n}$  &  0.005 & 0 &0 \\ 
					& & & $\widehat \eb^{aqLASSO}_{m+n}$& 0.005&0  & 0 \\  \cmidrule{3-7}
					& & $n^{3}$ & $\widehat \eb_{n}$  &  0.002 & 0 &0 \\ 
					& & & $\widehat \eb^{aqLASSO}_{m+n}$& 0.002&0  & 0 \\  \cmidrule{3-7}
					& & $n^{7/2}$ & $\widehat \eb_{n}$  &  0.0005 & 0 &0 \\ 
					& & & $\widehat \eb^{aqLASSO}_{m+n}$& 0.0005&0  & 0 \\  \midrule 
						100	& 115& $n^{5/2}$ & $\widehat \eb_{n}$  &  0.001 & 0 &0 \\ 
					& & & $\widehat \eb^{aqLASSO}_{m+n}$& 0.001&0  & 0 \\  \cmidrule{3-7}
					& & $n^{3}$ & $\widehat \eb_{n}$  &  0.0002 & 0 &0 \\ 
					& & & $\widehat \eb^{aqLASSO}_{m+n}$& 0.0001&0  & 0\\ [4pt] \hline 
					\end{tabular}
}
\end{center}
	\caption{Comparison study of $\widehat \eb_{n}$ and  $\widehat \eb^{aqLASSO}_{m+n}$,  when  $\varepsilon \sim {\cal N}(0,1)$. If $p \geq 5$ we take ${\cal S}^*=\{1, 2, 3\}$ with   $\beta_1^*=1$, $\beta^*_2=2$, $\beta^*_3=-1$. If $p = 3$ we take ${\cal S}^*=\{1\}$ with   $\beta_1^*=1$.  } 
		\label{Tabl5} 
\end{table*}

\begin{figure*}[h!] 
	\begin{tabular}{ccc}
		\includegraphics[width=0.3\linewidth,height=4.5cm]{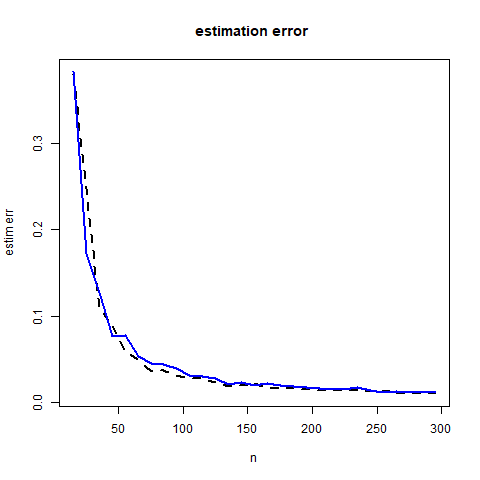} &
		\includegraphics[width=0.3\linewidth,height=4.5cm]{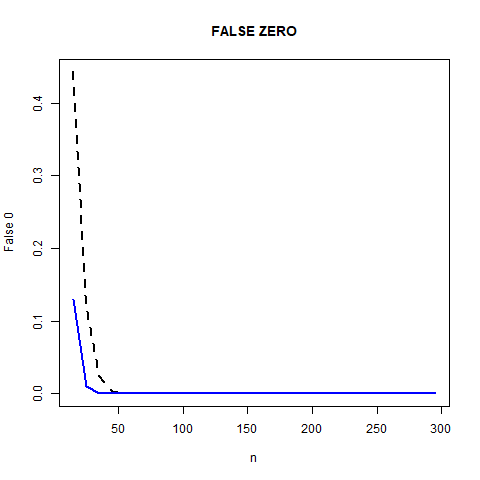}&
		\includegraphics[width=0.3\linewidth,height=4.5cm]{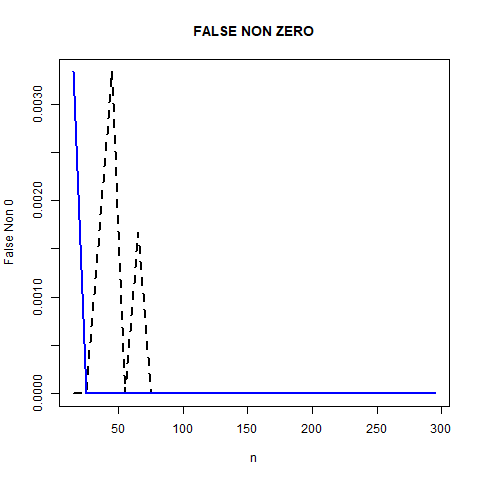} \\
		{\small (a) Evolution of estimation error.} &
		{\small (b)  Evolution of false  zeros.} &
		{\small (a) Evolution of false non-zeros.}
	\end{tabular}
	\caption{\scriptsize  Evolution with $n$  of estimations $\widehat \eb_{n}$  (\textit{transaLASSO}, black dashed), $\widehat \eb^{aqLASSO}_{m+n}$ (\textit{aLASSO\_all},  blue line), when  $\beta_1^*=0.1$, $\beta^2_0=0.2$, $\beta^*_3=-0.1$,  $p=10$, $\varepsilon \sim {\cal U}(-0.5,0.75)$. 	}
	\label{fig_Ubetapetit_p10}
\end{figure*}

\begin{figure*}[h!] 
	\begin{tabular}{ccc}
		\includegraphics[width=0.3\linewidth,height=4.5cm]{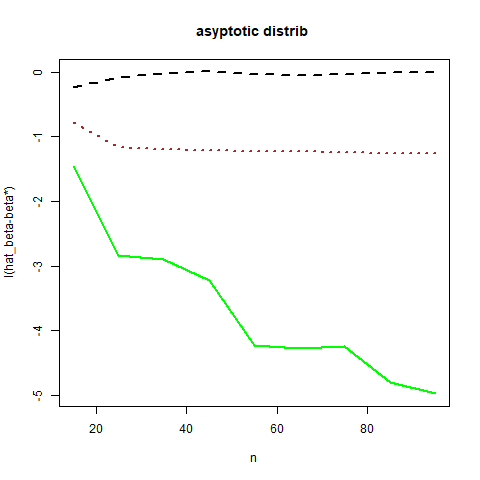} &
		\includegraphics[width=0.3\linewidth,height=4.5cm]{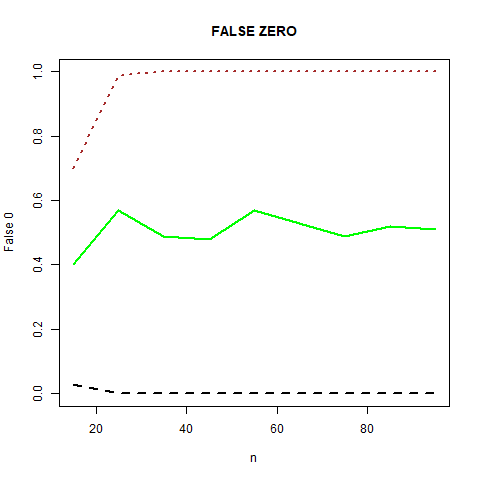}&
		\includegraphics[width=0.3\linewidth,height=4.5cm]{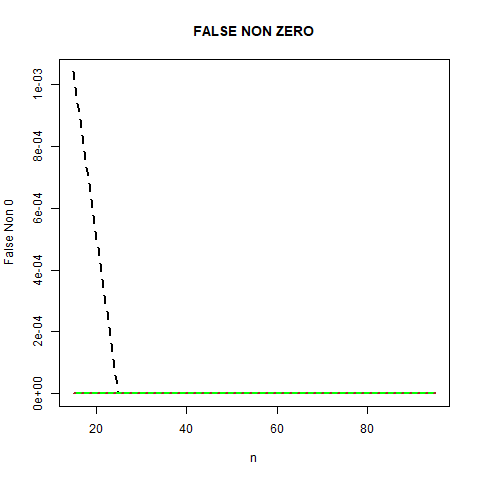} \\
		{\small (a) Evolution of $l(\widehat \eb_{n} -\eb^*)_{{\cal S}^*}$.} &
		{\small (b)  Evolution of false  zeros.} &
		{\small (a) Evolution of false non-zeros.}
	\end{tabular}
	\caption{\scriptsize  Evolution with $n$ of estimations $\widehat \eb_{n}$ for   three cases {\it (i)} (black dashed),{\it (ii)} (green line), {\it (iii)} (brown dotted line), when  ${\cal S}^*=\{1\}$, $\beta_1^*=1$,  $p=10$, $\varepsilon \sim 0.2 {\cal N}(0,1) +{\cal N}^2(0,1)$. 	}
	\label{fig_comptr1_p10}
\end{figure*}
\begin{figure*}[h!] 
\begin{tabular}{ccc}
	\includegraphics[width=0.3\linewidth,height=4.5cm]{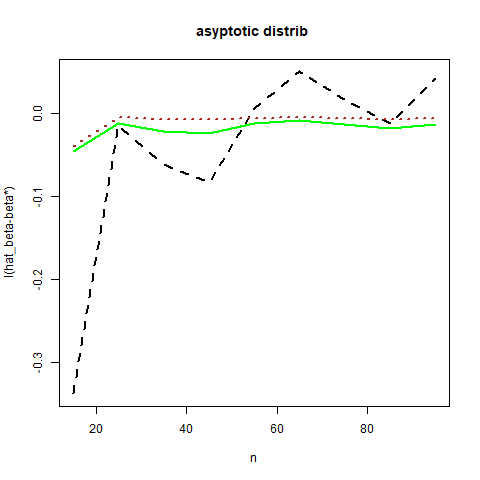} &
	\includegraphics[width=0.3\linewidth,height=4.5cm]{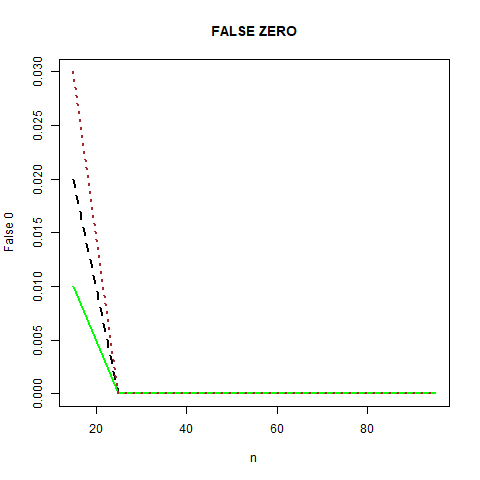}&
	\includegraphics[width=0.3\linewidth,height=4.5cm]{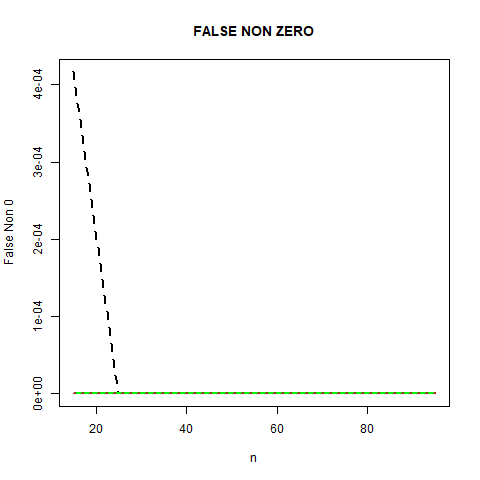} \\
	{\small (a) Evolution of $l(\widehat \eb_{n} -\eb^*)_{{\cal S}^*}$.} &
	{\small (b)  Evolution of false  zeros.} &
	{\small (a) Evolution of false non-zeros.}
\end{tabular}
\caption{\scriptsize  Evolution with $n$ of estimations $\widehat \eb_{n}$ for   three cases {\it (i)} (black dashed),{\it (ii)} (green line), {\it (iii)} (brown dotted line), when  ${\cal S}^*=\{1\}$, $\beta_1^*=1.125$,  $p=10$, $\varepsilon \sim 0.2 {\cal N}(0,1) +{\cal N}^2(0,1)$. 	}
\label{fig_comptr1_25_p10}
\end{figure*}
\begin{table*} 
		\begin{center}
						{\scriptsize
			\begin{tabular}{BBBBB}\hline  
		\multicolumn{1}{c}{	$(\beta^*_1,\beta^*_2)$}&	 \multicolumn{1}{c}{case}&  $  l\|(\widehat \eb_n -\eb^*)_{{\cal S}^*} \|_2 $  &    $\frac{| \{ j \in {\cal S}^*; \widehat \beta_{n,j} = 0\}|}{ |{\cal S}^*|}$  &   $\frac{| \{ j \in {{\cal S}^*}^c; \widehat \beta_{n,j} \neq 0\}|}{ |{{\cal S}^*}^c|}$\\ [2pt] \toprule
		 	(5,2)	  & (i)& 0.51 & 0& 0 \\[2pt]
			&  (ii) & 0.05 & 0& 0 \\[2pt]
			&  (iii) & 0.006 & 0& 0 \\[2pt] \cmidrule{1-5}
				(0.5,0.2)	  & (i)& 0.48 & 0& 0 \\[2pt]
			&  (ii) & 7 & 1& 0 \\[2pt]
			&  (iii) & 0.88 & 1& 0 \\ [2pt]\cmidrule{1-5}
				(5,0.2)	 & (i)& 0.52 & 0& 0 \\[2pt]
			&  (ii) & 2 & 0.5& 0 \\[2pt]
			&  (iii) & 0.25 & 0.5& 0 \\[2pt] \hline
			\end{tabular}
		}
	\end{center}
	\caption{ Comparison of  cases (i), (ii), (iii) on $\widehat \eb_n$,  when $p=10$, $n=100$,  ${\cal S}^*=\{1,2\}$, $\varepsilon \sim 0.2 {\cal N}(0,1) +{\cal N}^2(0,1)$. } 
	\label{Tabl4} 
\end{table*}
\textbf{Conclusion of this subsection}\\
The estimator $\widehat \eb^{aqLASSO}_{n}$ performs less well in terms of bias and variable selection when the number of observations in the target data is small. For coefficients close to 0 and asymmetric distribution errors, the estimator $\widehat \eb_n$ is less biased than $\widehat \eb^{aqLASSO}_{m+n}$.\\
The key advantage of the $\widehat \eb_n$ estimator is its computation time, which is significantly shorter than that of the other two adaptive LASSO quantile estimators, especially for high-dimensional models.
\subsection{Study of $\widehat \eb_{n}$}
\label{subsec_etudebn}
In this subsection,  our focus is solely on studying the estimator $\widehat \eb_{n}$ examining the three possible cases stated in Theorem \ref{TheoremTF 4.1}.  More specifically, in Table \ref{Tabl4} and Figures \ref{fig_comptr1_p10}, \ref{fig_comptr1_25_p10}, we study $\widehat \eb_{n}$ in  three cases (i), (ii), and (iii) of Theorem \ref{TheoremTF 4.1}.  The tuning parameters $\lambda_n$ and $\eta_n$ will differ for  three cases: for case (i), we take $\lambda_n=n^{1/2}$ and $\eta_n=n^{1/2+max(\gamma_1+0.001,(\gamma_1+\gamma_2)/2)}$; for case (ii)  $\lambda_n=\eta_n=1/2$ and for case (iii) $\lambda_n=n^{0.95}$, $\eta_n=n^{0.90}$.  The values of $l$ are $m^{1/2}$, $n^{1/2}$, and $n/\lambda_n$, respectively.  We refer to Table \ref{Tabl4}. For a model with  $p=10$ and only two non-zero coefficients, cases \textit{(ii)} and \textit{(iii)} yield more biased estimates than case  \textit{(i)}.  More importantly, they fail to detect coefficients close to zero. Moreover, the estimates of coefficients far from zero are slightly more biased in configuration \textit{(i)}. 
 The results for $p=100$ are not shown because they are identical to those obtained for $p=10$ and a fixed $\eb^*_{{\cal S}^*}$ and  the same distribution of $\varepsilon$. \\
 To better study the bias and the rate of convergence, for Figures \ref{fig_comptr1_p10} and \ref{fig_comptr1_25_p10} we consider only $\beta^*_1$ different from zero, with two possible values: 1 and 1.125. We set $p=10$, $\varepsilon \sim 0.2 {\cal N}(0,1) +{\cal N}^2(0,1)$, and varied $n$. 
 By comparing Figures \ref{fig_comptr1_p10} and \ref{fig_comptr1_25_p10}, we see that for $\beta^*_1 =1$, the estimator obtained in case (i) is unbiased.  This is unlike the estimators obtained in cases (ii) and (iii). It also produces very few false zeros in practice. When $|\beta^*_1|$ is large, then the estimators obtained in the three cases are quite similar (see Figure \ref{fig_comptr1_25_p10}).
 For the case shown in Figure \ref{fig_comptr1_25_p10}, with  $\beta_1^*=1.125$,   we have plotted the histograms of the estimates $\widehat \beta_{n,1}$. The Shapiro test confirms the normality of the estimations, yielding p-values of 0.22, 0.23, and 0.15 for (i), (ii), and (iii), respectively.\\  
 
\begin{figure*}[h!] 
	\begin{tabular}{ccc}
		\includegraphics[width=0.3\linewidth,height=4.5cm]{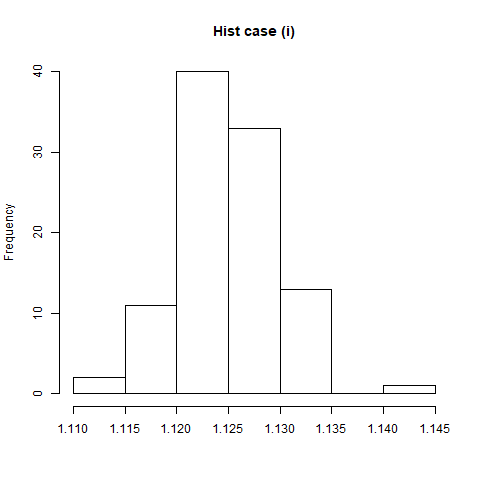} &
		\includegraphics[width=0.3\linewidth,height=4.5cm]{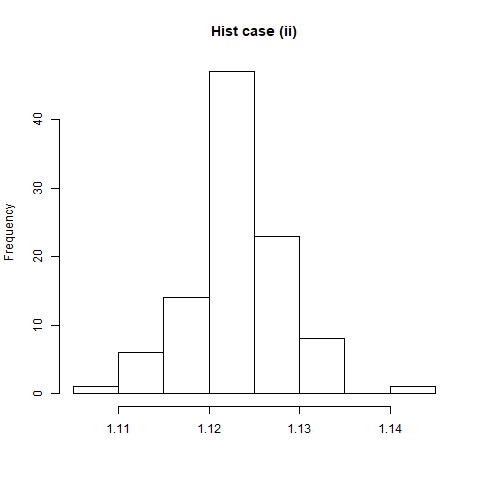}&
		\includegraphics[width=0.3\linewidth,height=4.5cm]{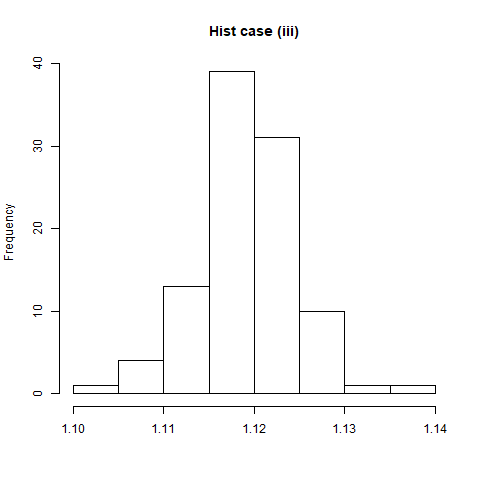} \\
		{\small (a) Case (i) }&
		{\small (b) Case (ii) }&
		{\small (a) Case (iii)}
	\end{tabular}
	\caption{\scriptsize  Histograms of  $\widehat \beta_{n,1}$ for   three cases {\it (i)},{\it (ii)}, {\it (iii)},    when $n=95$,  ${\cal S}^*=\{1\}$, $\beta_1^*=1.125$,  $p=10$, $\varepsilon \sim 0.2 {\cal N}(0,1) +{\cal N}^2(0,1)$. 	}
	\label{hist}
\end{figure*}
\textbf{Conclusion of this subsection}\\ 
For  three cases \textit{(i)}, \textit{(ii)}, and \textit{(iii)}, the simulations confirm the theoretical results stated in Section \ref{section_resultats}, namely that, for  $m \gg n$, it is preferable to choose tuning parameters  $\lambda_n$ and $\eta_n$ that satisfy the assumptions  of \textit{(ii)}.\\

\textbf{Overall conclusion based on the simulations}\\
The adaptive transfer LASSO quantile  method has an advantage over traditional adaptive quantile LASSO methods: it provides more accurate estimators, which are potentially of similar quality, in less computational time. This is especially true when the model has a large number of explanatory variables.  The proposed algorithm allows us to find an adaptive transfer LASSO estimator, $\widehat \eb_n$ that satisfies the sparsity property. This is achieved by using a consistent yet non-sparse estimator derived from the source training data, resulting in a shorter runtime. The estimator $\widehat \eb_n$ computed on the target data—which may be much smaller in size—utilizes the transfer of consistency from this initial estimator and achieves sparsity via an $L_1$-type penalty. Its computation time is significantly faster than that of the adaptive LASSO quantile  estimator computed on both the target and source data simultaneously. Furthermore, the result of the proposed algorithm does not depend on the initial value  $\eb^{(0)}$.
\section{Application on real data}
\label{section_application}
We illustrate our transfer learning estimation  method and algorithm on a real data application of physicochemical properties of protein tertiary structure. The data  can be downloaded from the \textit{Machine Learning Repository} site:\\
\href{http://archive.ics.uci.edu/ml/datasets/Physicochemical+Properties+of+Protein+Tertiary+Structure}{  {http://archive.ics.uci.edu/ml/datasets/Physicochemical+\\
		Properties+of+Protein+Tertiary+Structure}}\\
The data set consists of 45730 observations to predict the dependent variable $Y$ which represents the size of the residue. Hence, $m+n=45730$. The nine continuous explanatory variables are: $X_1$ - total surface area,  $X_2$ -  non polar exposed area, $X_3$ - fractional area of exposed non polar residue, $X_4$ - fractional area of exposed non polar part of residue, $X_5$ - molecular mass weighted exposed area, $X_6$ - average deviation from standard exposed area of residue, $X_7$ - euclidian distance, \textit{$X_8$} - secondary structure penalty, $X_9$ - spacial distribution constraints.  In order to remove the effect of the measurement unit, the nine explanatory variables  are standardized. Moreover, we consider the target data size  equal to $n=\lfloor\sqrt{45730}\rfloor=213$, with $\lfloor . \rfloor$ the integer part. Then the source observations are $1, \cdots, 45517$. The  quantile index $\tau$ will be estimated on the source data using  the   standardized  explained variable $(y_i-\bar y_m)/\widehat{\sigma}_y$, with $\widehat{\sigma}_Y$ the empirical standard deviation of $(y_i)_{1 \leqslant i \leqslant m}$.  Based on assumption (A1),  we consider an  empirical estimation of $\tau$ as:
	\[
	\widehat{\tau}=m^{-1}\sum^{m}_{i=1} \frac{y_i-\bar y_m}{\widehat{\sigma}_y} \e1_{y_i-\bar y_m<0}.
	\]
We obtain, $\widehat{\tau}=0.577$. The quantile estimation $\widetilde \eb_m$ of $(\beta_0^*, \beta^*_1, \cdots \beta^*_9)$ for the source model $Y_i=\beta_0+\beta_1 X_{1i}+\cdots +\beta_9 X_{9i} +\varepsilon_i$, $i=1, \cdots , m$, is $(8.24, 10.98, 0.48, 2.41, -7.49, -3.59, -1.82, -0.36,$ $ 1.08, -0.09)$. Then, for the  following studies, we will consider the response variable $\widetilde Y \equiv Y-8.24$ and the model without an intercept:
	\begin{equation}
	\widetilde Y_i=\beta_1 X_{1i}+\cdots +\beta_9 X_{9i} +\varepsilon_i.
	\label{eq_appli}
\end{equation}
	In Table \ref{Tabl_appli} we present the coefficient estimations   by three methods: $\widehat{\eb}_n$ - adaptive transfer LASSO, $\widehat \eb^{aqLASSO}_{m+n}$ - adaptive LASSO on source and target data, and  $\widehat \eb^{aqLASSO}_{n}$ - adaptive LASSO on target data. We present also the quantile estimation $\widetilde \eb_{m+n}$ for model (\ref{eq_appli}) calculated on source and target data. For a significance level of 0.05, the hypothesis tests for these last coefficient estimations indicates that only the coefficient for $X_9$ can be considered null (p-value=0.23). The same variable has the coefficient shrinked to  0  for  $\widehat{\eb}_n$, $\widehat \eb^{aqLASSO}_{m+n}$, $\widehat \eb^{aqLASSO}_{n}$. In return,  $\widehat \eb^{aqLASSO}_{m+n}$, $\widehat \eb^{aqLASSO}_{n}$ still have one, or two, other coefficients set to 0. For each of the four estimation, and then prediction, we also calculate  MAD=$n^{-1} \sum^{m+n}_{i=m+1} |\widetilde Y_i - \widehat{ \widetilde{ {Y}}}_i|$ the mean of  the absolute value of the model residuals  on target observations.\\
	The  predictions of $(\widetilde Y_i )_{m+1 \leqslant i \leqslant m+n}$ for target observations are equally accurate when using $\widehat{\eb}_n$, $\widehat \eb^{aqLASSO}_{m+n}$, $\widetilde \eb_{m+n}$ , while they are slightly more accurate when using $\widehat \eb^{aqLASSO}_{n}$. However, because the target number of observations is relatively small, the coefficient parameter estimations differ for  $\widehat \eb^{aqLASSO}_{n}$. 
	These data were also modeled by \cite{Ciuperca.2022} using   aggregation techniques of estimators. Depending on the number of sharing blocks, the paper of \cite{Ciuperca.2022} found that the aggregated method shrunk zero one through four coefficients, with coefficient of $X_9$ always set to 0. The results obtained in the present paper for $\widehat{\eb}_n$  are therefore consistent with those of \cite{Ciuperca.2022}.
 \begin{table*}[h] 
   	\begin{center}
 		{\scriptsize
 				\begin{tabular}{cccccccccccc}\hline  \toprule
 					estimator   & $\widehat {\cal {S}}^c$  & MAD & $\widehat{\beta}_1$ & $\widehat{\beta}_2$ & $\widehat{\beta}_3$ & $\widehat{\beta}_4$ & $\widehat{\beta}_5$ & $\widehat{\beta}_6$ & $\widehat{\beta}_7$ & $\widehat{\beta}_8$ & $\widehat{\beta}_9$ \\   \hline  \toprule
 					$\widehat{\eb}_n$ &   $\{9\}$ &  \textit{4.36}  & \textit{10.98} & \textit{0.48}& \textit{2.41} & -\textit{7.49} & \textit{-3.59} & \textit{-1.82}& \textit{-0.37} & \textit{1.08} & \textit{0 } \\ 
 					$\widehat \eb^{aqLASSO}_{m+n}$ & $\{2,9\}$ &4.36 &  10.08 & 0& 2.55 &-7.59 &   -2.57 &  -1.43 &  -0.20 &  1.08 &  0 \\ 
 					$\widehat \eb^{aqLASSO}_{n}$ & $\{3,5,9\}$ &4.14  & 6.07 & 5 .92 &  0 &  -6.36 &   0 &  -5.19 &  -2.66 & 2.39 & 0 \\ 
 					$\widetilde \eb_{m+n}$ & $\emptyset $ &4.36  & 10.92 &  0.49 &  2.41 &  -7.48 &   -3.54 &  -1.83 &  -0.36& 1.1 & -0.08 \\[4pt] \hline 
 				\end{tabular}
 		}
 	\end{center}
 	\caption{Results for application to the physicochemical properties of the tertiary structure of proteins: $n=213$, $m=45517$.  } 
 		\label{Tabl_appli} 
 \end{table*}
\section{Proofs of Theorems}
\label{section_proofs}
In this section we present the proofs of the two theorems stated in Section  \ref{section_resultats}.\\

	\noindent  {\bf Proof of Theorem \ref{TheoremTF 4.1}}.  \\
\textit{(i)} In this case, we consider  $l=m^{1/2}$.\\
We are studying the  difference:  $Z_{m,n}(\eb) -Z_{m,n}(\eb^*)=$
\begin{multline*}
{\cal P}_n+\frac{m^{(\gamma_1 - 1)/2} \lambda_n}{n} \sum_{j=1}^p \frac{|u_j+m^{1/2} \beta^*_j|-m^{1/2} |\beta^*_j|}{|z_{m,j}
	+m^{1/2} \beta^*_j|}\\+\frac{\eta_n}{n m^{(\gamma_2+1)/2}}\sum_{j=1}^p |z_{m,j}+m^{1/2} \beta^*_j|^{\gamma_2} \big(|u_j-z_{m,j}| - |z_{m,j}|\big),
\end{multline*}
with $Z_{m,n}(\eb)$ defined by (\ref{relZ}) and ${\cal P}_n$ by (\ref{eq3}).\\
Taking into account  relation (\ref{eq6}) obtained under assumptions (A1), (A2), (A3), we have:
\begin{equation}
	\label{ePn}
	\begin{split}
	\frac{n}{{m^{(\gamma_1-1)/2}} \lambda_n} {\cal P}_n  &=O_\PP \bigg(\frac{n m^{(1-\gamma_1)/2}}{\lambda_n}  l^{-2}\bigg)\\
	&=O_\PP \bigg(\frac{n m^{(1+\gamma_1)/2}}{\lambda_n}  \bigg).
	\end{split}
\end{equation}
We will now consider three subcases.\\
{\it (i)-(a)} If $m^{{\gamma_1}/2} \lambda_n \eta_n^{-1} \rightarrow 0$ as ${(m,n) \rightarrow \infty}$, let us consider:
\[
V_{m,n}(\eu) \equiv 	\frac{n}{{m^{(\gamma_1-1)/2}} \lambda_n}  \big(	Z_{m,n}(\eb) -	Z_{m,n}(\eb^*)\big).
\]
As in \cite{Takada.Fujisawa.24}, relations (184) and (185), taking into account relation (\ref{ePn}), we have:
\begin{equation*}
	\begin{split}
&	V_{m,n}(\eu)    \overset{\PP} {\underset{(m,n) \rightarrow \infty}{\longrightarrow}} \\
&	V(\eu) \equiv 
	\left\{
	\begin{array}{ll}
		\displaystyle{	\sum^p_{j=1}\frac{|u_j|}{|z_j|^{\gamma_1}} \e1_{\beta^*_j =0},} &\  \text{if} \quad \eu_{{\cal S}^*} =\ez_{{\cal S}^*},\\
		\infty,&   \text{otherwise} .
	\end{array}
	\right.
	\end{split}
\end{equation*} 
Then:
\begin{equation*}
	\begin{split}
	m^{1/2} (\widehat\eb_n-\eb^*)   =\argmin_{\eu \in \R^p} V_{m,n}(\eu)      \overset{\cal L} {\underset{(m,n) \rightarrow \infty}{\longrightarrow}}\\
	 \argmin_{\eu \in \R^p} 	V(\eu)=
	\left\{
	\begin{array}{ll}
		0, &\quad \text{if} \quad j \in {{\cal S}^*}^c\\
		z_j,& \quad \text{if} \quad j \in {{\cal S}^*} .
	\end{array}
	\right.
	\end{split}
\end{equation*} 
{\it (i)-(b)} If $m^{{\gamma_1}/2} \lambda_n \eta_n^{-1} \rightarrow \infty$ as ${(m,n) \rightarrow \infty}$, then,  taking into account relation (\ref{ePn}), we have:
\[
\begin{split}
V_{m,n}(\eu) \equiv 	\frac{n{m^{1/2}}}{ \eta_n}  \big(	Z_{m,n}(\eb) -	Z_{m,n}(\eb^*)\big)  \overset{\PP} {\underset{(m,n) \rightarrow \infty}{\longrightarrow}} \\
V(\eu) \equiv 
\left\{
\begin{array}{ll}
	\displaystyle{	\sum^p_{j=1}|\beta^*_j|^{\gamma_2}|u_j-z_{j}|,} &  \text{ if }   \eu_{{{\cal S}^*}^c} =\oo,\\
	\infty,&   \text{otherwise} .
\end{array}
\right.
\end{split}
\]
From where:
\begin{equation*}
	\begin{split}
	m^{1/2} (\widehat\eb_n-\eb^*)   =\argmin_{\eu \in \R^p} V_{m,n}(\eu)      \overset{\cal L} {\underset{(m,n) \rightarrow \infty}{\longrightarrow}} \\
	\argmin_{\eu \in \R^p} V(\eu)=
	\left\{
	\begin{array}{ll}
		0, &\quad \text{if} \quad j \in {{\cal S}^*}^c,\\
		z_j,& \quad \text{if} \quad j \in {{\cal S}^*} .
	\end{array}
	\right.
	\end{split}
\end{equation*}
For the last relation, we used the fact that the minimum value of $V_{m,n}(\eu)$ is 0, so we need to find the $\eu$ for which $V_{m,n}(\eu) = 0$. \\
{\it (i)-(c)} If $m^{{\gamma_1}/2} \lambda_n \eta_n^{-1} \rightarrow \rho_0 >0$ as ${(m,n) \rightarrow \infty}$, then, taking into account relation (\ref{ePn}), we have:
\[
\begin{split}
V_{m,n}(\eu) \equiv 	\frac{n{m^{1/2}}}{ \eta_n}  \big(	Z_{m,n}(\eb) -	Z_{m,n}(\eb^*)\big)  \overset{\PP} {\underset{(m,n) \rightarrow \infty}{\longrightarrow}}\\
 V(\eu) \equiv 
\sum^p_{j=1}\big(|\beta^*_j|^{\gamma_2}|u_j-z_j| \e1_{\beta^*_j \neq 0}+r_0 \frac{|u_j|}{|z_j|^{\gamma_1}} \e1_{\beta^*_j =0}\big).
\end{split}
\]
From where we have:
\begin{equation*}
	\begin{split}
m^{1/2} (\widehat\eb_n-\eb^*)   =\argmin_{\eu \in \R^p} V_{m,n}(\eu)      \overset{\cal L} {\underset{(m,n) \rightarrow \infty}{\longrightarrow}} \\
	\argmin_{\eu \in \R^p} V(\eu)=
	\left\{
	\begin{array}{ll}
		0, &\quad \text{if} \quad j \in {{\cal S}^*}^c,\\
		z_j,& \quad \text{if} \quad j \in {{\cal S}^*} .
	\end{array}
	\right.
	\end{split}
\end{equation*}
For the last relation, we used the fact that the minimum of $V_{m,n}(\eu)$ is 0; therefore, we must find the $\eu$ for which $V_{m,n}(\eu) = 0$. \\
Then, the convergence of relation  (\ref{eqTF30}) is proven for all three subcases.\\
\textit{(ii)}  In this case, we consider $l = n^{1/2}$.\\
Taking into account the results proven in case \textit{(i)} and the proof of Theorem 4.1 of \cite{Takada.Fujisawa.24}, we obtain   $V_{m,n}(\eu) \equiv     n  \big(    Z_{m,n}(\eb) -	Z_{m,n}(\eb^*)\big)  \overset{\PP} {\underset{(m,n) \rightarrow \infty}{\longrightarrow}} V(\eu)$,  with $V(\eu)  $ defined 
\begin{itemize}
	\item when 	 $\eu_{{{\cal S}^*}^c}=\oo$ by 
	\begin{align*}
		&	\frac{f(0)}{2}\eu^\top \eU \eu + \eu^\top \bfW \\
			&+	\sum_{j \in {{\cal S}^*}}\bigg(\lambda_0 \frac{\sgn(\beta^*_j)}{|\beta^*_j|^{\gamma_1}}u_j+\eta_0|\beta^*_j|^{\gamma_2}|u_j-r_0^{1/2}z_j|\bigg),
	\end{align*}
\item  by $\infty$ in other cases.
\end{itemize}
\textit{(iii)} In this case, we consider $l = n \lambda_n^{-1}$.\\
Let $V_{m,n}(\eu) \equiv n^2 \lambda^{-2}_n \left( Z_{m,n}(\eb) - Z_{m,n}(\eb^*) \right)$. We can write:
\begin{align*}
	V_{m,n}(\eu) & =\frac{n^2}{\lambda_n^2} {\cal P}_n +\sum^p_{j=1} \bigg(\frac{\big|u_j+\frac{n}{\lambda_n} \beta^*_j\big|- \big|\frac{n}{\lambda_n} \beta^*_j \big|}{\big|\beta^*_j+m^{-1/2} z_j\big|^{\gamma_1}} \e1_{\beta^*_j \neq 0}\\ &\qquad+\frac{m^{\gamma_1/2}|u_j|}{|z_j|^{\gamma_1}} \e1_{\beta^*_j =0}\bigg) \\
	&   + \frac{\eta_n}{\lambda_n}\sum^p_{j=1}  \bigg(\bigg|\beta^*_j+\frac{z_j}{m^{1/2}}\bigg| \e1_{\beta^*_j \neq 0} +\frac{|z_j|^{\gamma_2}}{m^{{\gamma_2}/2}} \e1_{\beta^*_j=0}\bigg)\\ &\bigg(\bigg|u_j-\frac{n}{m^{1/2} \lambda_n}z_j\bigg| - \bigg| \frac{n}{m^{1/2} \lambda_n}z_j\bigg| 	 \bigg)\\
	& = \frac{f(0)}{2} \eu^\top \bigg( \frac{1}{n}\sum^{m+n}_{i=m+1} \eX_i \eX_i^\top\bigg) \eu +l^2 \frac{1}{n} \bfW_n \eu \\
	&+\sum^p_{j=1} \bigg(\frac{\big|u_j+\frac{n}{\lambda_n} \beta^*_j\big|- \big|\frac{n}{\lambda_n} \beta^*_j \big|}{\big|\beta^*_j+m^{-1/2} z_j\big|^{\gamma_1}} \e1_{\beta^*_j \neq 0}\\
	&+ \frac{m^{\gamma_1/2}|u_j|}{|z_j|^{\gamma_1}} \e1_{\beta^*_j =0}\bigg) \\
	&  + \frac{\eta_n}{\lambda_n}\sum^p_{j=1}  \bigg(\bigg|\beta^*_j+\frac{z_j}{m^{1/2}}\bigg| \e1_{\beta^*_j \neq 0} +\frac{|z_j|^{\gamma_2}}{m^{{\gamma_2}/2}} \e1_{\beta^*_j=0}\bigg) \\
	&\bigg(\bigg|u_j-\frac{n}{m^{1/2} \lambda_n}z_j\bigg| - \bigg| \frac{n}{m^{1/2} \lambda_n}z_j\bigg| 	 \bigg),
\end{align*}
which converges in probability as $(m,n) \rightarrow \infty$ to $V(\eu)  $, with $V(\eu)  $ defined 
\begin{itemize}
	\item when 	 $\eu_{{{\cal S}^*}^c}=\oo$ by 
	\begin{align*}
	\frac{f(0)}{2}\eu^\top \eU \eu +\sum_{j \in {{\cal S}^*}} \frac{\sgn(\beta^*_j)}{|\beta^*_j|^{\gamma_1}}u_j,
	\end{align*}
	\item  by $\infty$ in other cases.
\end{itemize}
Then, we deduce that: 
\begin{equation*}
	\begin{split}
&	\frac{	n }{\lambda_n} (\widehat\eb_n-\eb^*)        \overset{\cal L} {\underset{(m,n) \rightarrow \infty}{\longrightarrow}} \argmin_{\eu \in \R^p} V(\eu)\\
&=  \argmin_{\eu \in  {\cal U}}
	\bigg(	\frac{f(0)}{2}\eu^\top \eU \eu +\sum_{j \in {{\cal S}^*}} \frac{\sgn(\beta^*_j)}{|\beta^*_j|^{\gamma_1}}u_j \bigg),
\end{split}
\end{equation*}
with $ {\cal U} \equiv \{ \eu ;  \eu_{{{\cal S}^*}^c}=\oo\}$.
\hspace*{\fill}$\blacksquare$ \\

	\noindent  {\bf Proof of Theorem \ref{TheoremTF_4.3}}.  \\
Let’s first state the KKT optimality conditions:\\
(a) $\forall j \in \widehat{\cal S}_{m,n}$  and  $\widehat\beta_{n,j}  \neq \widetilde\beta_{m,j}$, the following equality holds with probability one: 
\begin{equation}
	\begin{split}
& \tau \sum^{m+n}_{i=m+1} X_{ji} - \sum^{m+n}_{i=m+1} X_{ji} \e1_{Y_i < \eX_i^\top \widehat \eb_{n}}= \\
& \lambda_n v_{m,j} \sgn(\widehat \beta_{n,j}) +\eta_n \omega_{m,j} \sgn(\widehat \beta_{n,j}- \widetilde \beta_{m,j}) .
 \end{split}
 \label{eq7a}
\end{equation}
(b) $ \forall j \in \widehat{\cal S}_{m,n}$  and  $\widehat\beta_{n,j} =\widetilde\beta_{m,j}$, the following equality holds with probability one: 
\begin{equation}
	\begin{split}
&	\tau \sum^{m+n}_{i=m+1} X_{ji} - \sum^{m+n}_{i=m+1} X_{ji} \e1_{Y_i < \eX_i^\top \widehat \eb_{n}} \quad \\
& \qquad =\lambda_n v_{m,j} \sgn(\widehat \beta_{n,j}).
	\label{eq7b}
	\end{split}
\end{equation}
(c) $ \forall j \not \in \widehat{\cal S}_{m,n} $, the following inequality holds with probability one: 
\begin{equation}
	\begin{split}
		& \bigg|\tau \sum^{m+n}_{i=m+1} X_{ji} - \sum^{m+n}_{i=m+1} X_{ji} \e1_{Y_i < \eX_i^\top \widehat \eb_{n}}  \bigg| \\
		& \qquad \leq \lambda_n v_{m,j}   +\eta_n \omega_{m,j}.
		\label{eq7c}
	\end{split}
\end{equation}
Since by Theorem \ref{TheoremTF 4.1} we have $\widehat \eb_{n}  \overset{\PP} {\underset{(m,n) \rightarrow \infty}{\longrightarrow}} \eb^*$,
then 
\begin{equation}
\label{ShSa}
\PP[ {\cal S}^* \subseteq\widehat{\cal S}_{m,n}] {\underset{(m,n) \rightarrow \infty}{\longrightarrow}}  1 .
\end{equation}
Let $j \in \widehat{\cal S}_{m,n} \cap {{\cal S}^*} ^c $. Then $\beta^*_{j}=0$, $\widehat \beta_{n,j} \neq 0$ and  $\widehat \beta_{n,j} \overset{\PP}{\underset{(m,n) \rightarrow \infty}{\longrightarrow}} 0$ with the convergence rate $l^{-1}$.\\
We will show that, 
\begin{equation}
	\label{jhS}
	\PP[j \in \widehat{\cal S}_{m,n} \cap {{\cal S}^*}^c ] {\underset{(m,n) \rightarrow \infty}{\longrightarrow}}  0
\end{equation}
and therefore 
\begin{equation}
	\label{ShSb}
	\PP[\widehat{\cal S}_{m,n} \subseteq {\cal S}^*] \rightarrow 1, \quad \text{as } (m,n) \rightarrow \infty .
\end{equation}
Relations (\ref{ShSa}) and (\ref{ShSb}) imply  $\PP[\widehat{\cal S}_{m,n} = {\cal S}^*] \rightarrow 1$ as $(m,n) \rightarrow \infty$,  that is, the statement of this theorem. 
First, the left-hand side of  KKT relations (\ref{eq7a})–(\ref{eq7c}) can be written as: 
\begin{align}
&	\frac{1}{n}   \sum^{m+n}_{i=m+1}  \big( \tau- \e1_{Y_i < \eX_i^\top \widehat \eb_{n}}\big)  X_{ji}\nonumber \\
& =\frac{1}{n}   \sum^{m+n}_{i=m+1} \big(\tau - \e1_{\varepsilon_i< 0}\big) X_{ji} \nonumber \\
&\quad+ \frac{1}{n}   \sum^{m+n}_{i=m+1} \big( \e1_{\varepsilon_i<0} -\e1_{\varepsilon_i < l^{-1} \eX_i^\top \eu}\big) X_{ji} \nonumber  \\ 
	& \equiv T_1+T_2.
	\label{ea14}
\end{align}
We will study $T_1$, $T_2$, the right-hand sides of equations (\ref{eq7a})–(\ref{eq7c}) under the assumption that $j \in \ \widehat{\cal S}_{m,n} \cap {{\cal S}^*} ^c$, and we will compare them with equation (\ref{ea14}).
By CLT for independent random variable sequences, under assumptions (A1) and (A2), we have $T_1=O_\PP(n^{-1/2})$. \\
For the term $T_2$, we first compute $\eE[T_2]$. Using the fact that as $t \rightarrow 0$, we have $F(t)-F(0)= t f(0)+o(t)$. Then we can write:
\begin{align*}
	\eE[T_2] & = n^{-1} \sum^{m+n}_{i=m+1} \big(F(0) - F(l^{-1}\eX_i^\top \eu)\big)X_{ji}\\
	&= - n^{-1}l^{-1} \sum^{m+n}_{i=m+1}  f(0) \eX_i^\top \eu X_{ji}(1+o(1))\\
	&=O(l^{-1}).
\end{align*}  
In the last derivation, we used the fact that the vectors $\eX_i$ and $\eu$ are bounded.  Let us now calculate the variance of $T_2$, for which we have $\Var[T_2] \leq \eE[T_2^2]$. For this last expectation, using the independence of the $\varepsilon_i$ we have:
\begin{align*}
	\eE[T_2^2] & = n^{-2} \sum^{m+n}_{i=m+1} \eE \big[ \big(\e1_{\varepsilon_i<0}- \e1_{\varepsilon_i < l^{-1} \eX_i^\top \eu}\big)^2 X^2_{ji}\big] \\
	& = n^{-2} \sum^{m+n}_{i=m+1} \bigg( \eE [ \e1_{\varepsilon_i <0} ] +\eE[\e1_{\varepsilon_i < l^{-1} \eX_i^\top \eu}] \\
	&\quad -2 \eE \big[\e1_{\varepsilon_i < \min(0, l^{-1}\eX_i^\top \eu)}\big]\bigg) X^2_{ji}.
\end{align*}
Without loss of generality, we assume that $0 < l^{-1}\eX_i^\top \eu$. The case $0 > l^{-1}\eX_i^\top \eu$ is handled in the same way. Then, using assumptions (A1), (A2) together with the fact $\max_{m+1 \leqslant i \leqslant m+n} \| \eX_i \|_2$ is bounded, we obtain:
\begin{align*}
	\eE[T_2^2] & = n^{-2} \sum^{m+n}_{i=m+1}  \big( F(0)- F(l^{-1}\eX_i^\top \eu)- 2 F(0)\big)X^2_{ji} \\
	& = n^{-2} \sum^{m+n}_{i=m+1} \big(F(l^{-1}\eX_i^\top \eu) -F(0)\big)X^2_{ji} \\
	& = n^{-2} \sum^{m+n}_{i=m+1} f(0)l^{-1}\eX_i^\top \eu X^2_{ji} (1+o(1))\\
	& =O(l^{-1}n^{-1}).
\end{align*}
By the Bienaymé–Chebyshev inequality, we then have:
\begin{equation*}
	T_2=O_\PP(\eE[T_2])=O_\PP(l^{-1}).
\end{equation*}
We obtained the following for relation (\ref{ea14}):
\begin{equation}
	\label{g_KKT}
	\begin{split}
&	\frac{1}{n}   \sum^{m+n}_{i=m+1}  \big( \tau- \e1_{Y_i < \eX_i^\top \widehat \eb_{n}}\big)  X_{ji}\\
 & \quad=O_\PP(n^{-1/2})+O_\PP(l^{-1}).
\end{split}
\end{equation}
On the other hand, 
\begin{equation}
	\label{dr1_KKT}
	\begin{split}
	\frac{\lambda_n v_{m,j}}{n}
	&=\frac{\lambda_n m^{{\gamma_1}/2}}{n}\frac{1}{|m^{1/2}\widetilde \beta_{m,j}|^{\gamma_1}} \\
&=O_\PP \big( \lambda_n m^{{\gamma_1}/2} n^{-1} \big).
	\end{split}
\end{equation}
We have used the fact that $|m^{1/2}\widetilde \beta_{m,j}|^{\gamma_1}$ is bounded with probability approaching 1 as $n \rightarrow \infty$.\\
We study:
\begin{equation}
	\label{ra}
	\begin{split}
	&B \equiv \frac{\lambda_n v_{m,j}}{n} \cdot \\
	&\cdot  \frac{n}{\tau \sum^{m+n}_{i=m+1} X_{ji} - \sum^{m+n}_{i=m+1} X_{ji} \e1_{Y_i < \eX_i^\top \widehat \eb_{n}} }.
	\end{split}
\end{equation}
We will continue the proof of this theorem by distinguishing between the three cases considered in Theorem \ref{TheoremTF 4.1} for the sequence $l$.\\
In  case {\it (i)}, we have $l=m^{1/2}$. Since $r_0 \geq 0$, then $T_1+T_2=O_\PP(n^{-1/2})$.  The term $B$ in relation (\ref{ra}) becomes:
\[
\begin{split}
O_\PP\big(\lambda_n m^{\gamma_1/2} \frac{n^{-1}}{n^{-1/2}}\big)&=O_\PP\big(\lambda_n m^{\gamma_1/2}n^{-1/2}\big)\\ &{\underset{(m,n) \rightarrow \infty}{\longrightarrow}} \infty .
\end{split}
\]
If  $\lambda_n m^{(\gamma_1+1)/2}n^{-1} \longrightarrow \infty$,  then $B =O_\PP \big(\lambda_n m^{\gamma_1/2} n^{-1} m^{1/2}\big) {\underset{(m,n) \rightarrow \infty}{\longrightarrow}} \infty$
Hence,  $ \lambda_n v_{m,j}n^{-1}$ is much larger than  $\left| \tau \sum^{m+n}_{i=m+1} X_{ji} - \sum^{m+n}_{i=m+1} X_{ji} \e1_{Y_i < \eX_i^\top \widehat \eb_{n}}\right|$, which contradicts   relations (\ref{eq7b}) and (\ref{g_KKT}).\\
If $\widehat \beta_{n,j} \neq \widetilde \beta_{m,j} $,  we compare  $\lambda_n v_{m,j} \sgn(\widehat \beta_{n,j})$ and $\eta_n \omega_{m,j} \sgn(\widehat \beta_{n,j} - \widetilde \beta_{m,j})$. Taking into account that $\widetilde \beta_{m,j} =m^{-1/2}$, we have:
\begin{equation*}
	\label{eq8}
	\begin{split}
	\frac{\lambda_n v_{m,j}}{\eta_n \omega_{m,j}} &= \frac{\lambda_n  m^{{ \gamma_1}/2}}{\eta_n m^{-{  \gamma_2}/2}}\\
	&=\frac{\lambda_n}{\eta_n}m^{(\gamma_1+\gamma_2)/2} {\underset{(m,n) \rightarrow \infty}{\longrightarrow}} \infty,
	\end{split}
\end{equation*}
and thus, the absolute value of the right-hand side of (\ref{eq7a}) converges in probability to $+\infty$, while the left-hand side is bounded.\\
Thus, if $l=m^{1/2}$, then relation (\ref{jhS}) is proven.\\
We consider case {\it (ii)}, which occurs for $l=n^{1/2}$. 
Then, the  left-hand side of (\ref{eq7a}) and (\ref{eq7b}) expressed as (\ref{g_KKT}) is: $T_1+T_2=O_\PP(n^{-1/2})+O_\PP(n^{-1/2})=O_\PP(n^{-1/2})$. On the other hand, for the first term on the right-hand side of (\ref{eq7a}) and (\ref{eq7b}) we have:
\[
\frac{\lambda_n v_{m,j}}{n}=O_\PP \bigg( \frac{\lambda_n}{n} m^{{\gamma_1}/2} \bigg)
\]
and for the second term on the right-hand side of (\ref{eq7a})  we have  $\eta_n \omega_{m,j}=O_\PP \big( \eta_n m^{{\gamma_2}/2} \big)$. Then, we deduce that:
\[
\begin{split}
\frac{\lambda_n v_{m,j}}{\eta_n \omega_{m,j}}&=O_\PP \bigg(\frac{\lambda_n}{\eta_n}m^{{(\gamma_1+\gamma_2)}/2}\bigg)
\end{split}
\]
\[
\begin{split}
& =O_\PP \bigg(\lambda_n \sqrt{\frac{m^{{\gamma_1}/2}}{n}} \frac{\sqrt{n}}{\eta_n} m^{{\gamma_2}/2} \bigg) \\
&{\underset{(m,n) \rightarrow \infty}{\longrightarrow}} \infty \cdot \infty =\infty .
\end{split}
\]
So, it is the first term $\lambda_n v_{m,j}$ that matters.\\
We compare the left-hand sides of (\ref{eq7a}) and (\ref{eq7b}) with $\lambda_n v_{m,j}$:
\[
\frac{n^{-1/2}}{\lambda_n v_{m,j}} n =\frac{n^{1/2}}{\lambda_n m^{{\gamma_1}/2}} {\underset{(m,n) \rightarrow \infty}{\longrightarrow}}  0.
\]
Therefore, the left-hand side and the right-hand side cannot be equal in (\ref{eq7a}) and (\ref{eq7b}). This implies relation (\ref{jhS}).\\
We consider case {\it (iii)}, which occurs when $l=n/\lambda_n$. In this case, since $n^{-1/2} \lambda_n {\underset{n \rightarrow \infty}{\longrightarrow}} \infty$, we have: $T_1+T_2=O_\PP(n^{-1/2})+O_\PP(n^{-1}\lambda_n)=O_\PP(n^{-1} \lambda_n)$. We compare the right-hand sides of equations (\ref{eq7a}) and (\ref{eq7b}):  $\lambda_n v_{m,j}(\eta_n \omega_{m,j})^{-1}=\lambda_n \eta_n^{-1}m^{(\gamma_1+\gamma_2)/2} {\underset{(m,n) \rightarrow \infty}{\longrightarrow}} \infty \cdot \infty  =\infty$ and therefore it is $\lambda_n v_{m,j} \sgn(\widehat \beta_{n,j})$ that dominates $\eta_n \omega_{m,j} \sgn(\widehat \beta_{n,j}- \widetilde \beta_{m,j})$. We compare the left and right sides of the KKT relations. Taking into account relation (\ref{dr1_KKT}), the term $B$ in (\ref{ra}) is:  
\[
\frac{\lambda_n m^{\gamma_1/2}n^{-1}}{l^{-1}}= \lambda_n m^{\gamma_1/2}n^{-1} \frac {n}{\lambda_n}=m^{\gamma_1/2} {\underset{(m,n) \rightarrow \infty}{\longrightarrow}}\infty.
\]
Thus, in relations (\ref{eq7a}) and (\ref{eq7b}), the left-hand side is bounded, but the right-hand side is not. Therefore,    relation (\ref{jhS}) is proven. \\
The proof of the theorem is thus complete.\\
\hspace*{\fill}$\blacksquare$  \\

\end{document}